\documentclass[onecolumn,10pt]{IEEEtran}
\usepackage{graphicx}
\usepackage{amsmath,amssymb}
\usepackage{cases}
\usepackage{color}
\usepackage{amsfonts}
\usepackage[latin9]{inputenc}
\usepackage{indentfirst}
\usepackage{cite}
\usepackage{caption}
\usepackage{algorithm}
\usepackage{algpseudocode}
\usepackage{booktabs}
\usepackage{blkarray}
\usepackage{subfigure}
\usepackage{multirow}
\usepackage{amsthm}
\usepackage{diagbox}
\makeatletter
\newtheoremstyle{mythm}{3pt}{3pt}{}{16pt}{\bfseries}{:}{.5em}{}
\theoremstyle{mythm}
\newtheorem{theorem}{Theorem}
\newtheorem{example}{Example}
\newtheorem{definition}{Definition}
\newtheorem{remark}{Remark}

\newtheorem{corollary}{Corollary}
\newtheorem{lemma}{Lemma}
\newtheorem{construction}{Construction}

\begin{document}
\title{Centralized Hierarchical Coded Caching Scheme
over Two-Layer Networks
\author{Yun Kong, Youlong Wu, and Minquan Cheng}
\thanks{Yun Kong and Minquan Cheng are with Guangxi Key Lab of Multi-source Information Mining $\&$ Security, Guangxi Normal University,
Guilin 541004, China.  (e-mail: yunkong2022@outlook.com, chengqinshi@hotmail.com). }
\thanks{Youlong Wu is with the School of Information Science and Technology, ShanghaiTech University,
201210 Shanghai, China. (e-mail: wuyl1@shanghaitech.edu.cn).}
}
\maketitle
%

\begin{abstract}  This  paper considers a hierarchical caching system where a server
connects with multiple mirror sites, each connecting with a distinct set of users, and both the mirror sites and users are equipped with caching memories.   Although there already exist	 works studying this setup and proposing coded caching scheme to reduce   transmission  loads, two main problems are remained to address: 1)  the optimal communication load  under the uncoded placement  for the first hop, denoted by $R_1$, is still unknown.  2) the previous schemes are based on Maddah-Ali and Niesen's data placement and delivery, which requires high    subpacketization level. How to achieve the well tradeoff between    transmission loads and subpacketization level for the hierarchical caching system is unclear.  In this paper, we aim to address these two problems. We first propose a new combination structure named hierarchical placement delivery array (HPDA), which  characterizes the data placement and delivery for any hierarchical caching system. Then we construct two classes of HPDAs, where the first class leads to a scheme achieving the optimal $R_1$ for some cases, and the second class requires a smaller  subpacketization level at the cost of slightly increasing   transmission loads.

\end{abstract}

\begin{IEEEkeywords}
hierarchical placement delivery array, hierarchical coded caching scheme, transmission load, subpacketization.
\end{IEEEkeywords}
\section{Introduction}
With the growing data demand especially the streaming media, there exists an extreme transmission pressure during the peak traffic hours in the wireless network. It is well known that caching system is an efficient way to reduce transmission during the peak traffic hours by shifting traffic from peak to off peak hours. That is, the central server can firstly place some contents into users' memories during the off peak traffic hours. During the peak traffic hours, the central server would only transmit the contents which have not been cached by the users. In addition, Maddah-Ali and Niesen in \cite{MN} showed that the contents cached by the users can be used to further reduce the transmission load during the peak traffic hours since these contents could generate more multicast opportunities among the users. They first introduced the centralized $(K,M,N)$ caching system where a single server having access to a library containing $N$ files is connected to $K$ cache-aided users whose cache size is $M$ files through an error-free shared-link. An $F$-division $(K,M,N)$ coded caching scheme contains two phases. During the placement phase, each file is divided into $F$ packets, where $F$ is referred as the subpacketization, and the server places at most $MF$ packets into each user's cache without any information about   users' demands. Since the packets are directly cached by the users, this placement strategy is called uncoded placement. In the delivery phase, each user requests a file from the server randomly and the server broadcasts some coded messages to users such that all the users can decode their requesting files with the help of their cached packets. We focus on the  worst case where each user requests a distinct file. The transmission amount normalized by the  size of file is defined as the ``transmission load".

To further reduce the transmission load of a $(K,M,N)$ caching system, Maddah-Ali and Niesen proposed the first centralized coded caching scheme \cite{MN} (MN scheme) and the first decentralized coded caching scheme \cite{MND} (MN decentralized scheme) respectively.
It is worth noting that when $N\geq K$, MN scheme has the minimum transmission load under uncoded placement in \cite{WTDPP}.  However, the subpacketization $F$ of the MN scheme increases exponentially with the growing on user number $K$. In order to design a scheme with low subpacketization, the author in \cite{YCTC} proposed a combination structure named placement delivery array (PDA) to simultaneously characterize the placement and delivery phase of a coded caching scheme. The MN scheme could also be depicted by PDA referred to as MN PDA. It is worth noting that MN scheme has the minimum subpacketization for the fixed minimum transmission load among all the schemes which can be realized by PDAs \cite{CJTY}. Besides PDA, there also exists some other combination structures to describe a coded caching scheme aiming at reducing the subpacketization level, such as \cite{SCZYG,YTCC,CLTW}, etc.


\subsection{Two-layer hierarchical network model}
\label{subsec:network-model}
\begin{figure}[http!]
\centering
\includegraphics[scale=0.5]{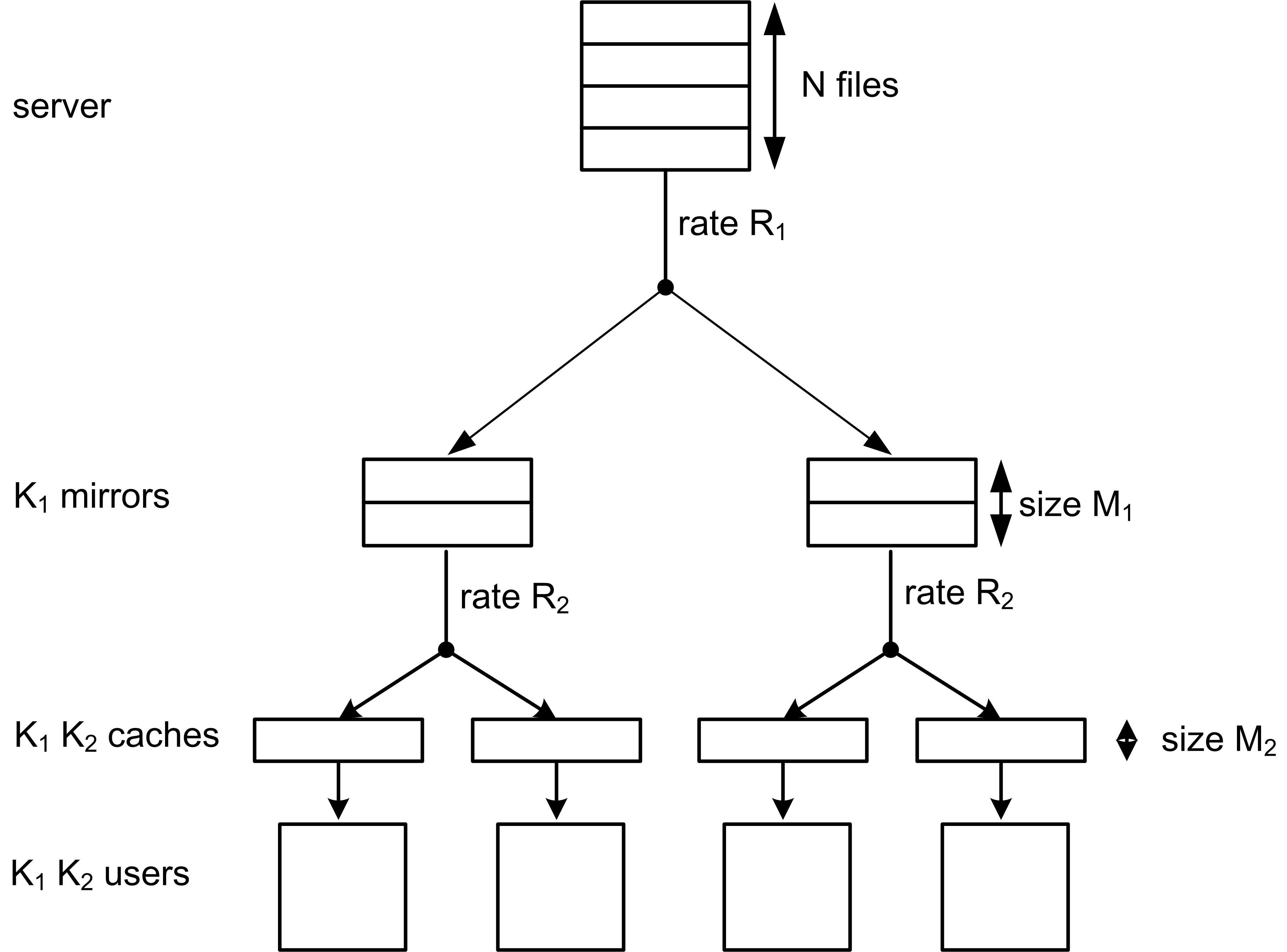}
\caption{The $(K_1,K_2;M_1,M_2;N)$ hierarchical  caching system with $N=4$, $K_1=K_2=2$, $M_1=2$ and $M_2=1$.}
\label{fig-model}
\end{figure}
In reality, the caching systems consist multiple layers of caches in most cases, which means between the central server and the end users, there are some in-between devices, and all these devices are arranged in a tree-like hierarchy with the central server at the root node while the end users act as the leaf nodes. Between each layer, the parent node communicates with its children nodes. As illustrated in Fig. \ref{fig-model}, a $(K_1$, $K_2$; $M_1$, $M_2$; $N)$ hierarchical caching system was first studied in \cite{KNMD}. That is, a two-layer hierarchical network consists of a single origin server and $K_1$ cache-aided mirror sites and $K_1K_2$ cache-aided users where the server hosts a collection of $N\geq K_1K_2$ files with equal size, each mirror site and each user have memories of size $M_1$ files and $M_2$ files respectively where $M_1$, $M_2\leq N$. The server is connected through an error-free shared link to each mirror site. Each mirror site is connected through an error-free broadcast link to $K_2$ users and each user is connected to only one mirror site.

Due to the complexity of the above hierarchical caching system, there are only a few studies in \cite{KNMD,ZWXWL,WWCY}.
In fact, all the existing schemes consist of two-subsystem model controlled by two parameters $\alpha,\beta\in[0,1]$, where the first subsystem includes the whole cache memory size of each mirror sites, an $\alpha$ fraction of each file and a $\beta$ fraction of users' cache memory size, and the second subsystem includes the rest $1-\alpha$ fraction of each file and a $1-\beta$ fraction of users' cache memory size. In \cite{KNMD}, the authors directly used the MN decentralized scheme in each layer. Then the requested files are built in each layer, so it results in a high transmission load of $R_1$. This scheme is referred as KNMD scheme. Without building the whole file in each layer, the authors in \cite{ZWXWL} proposed a hybrid scheme ( the ZWXWL scheme) based on two MN schemes for the two layers, so it has a lower transmission load of $R_1$ and its $R_2$ achieves the minimum transmission load under uncoded placement. However, the ZWXWL scheme ignores the usefulness of users' cache in the second layer when decoding the messages sent from the server. The author in \cite{WWCY} proposed an improved scheme (the WWCY scheme) by utilize two MN schemes. the WWCY scheme also achieves the minimum load of $R_2$ under uncoded placement, while the load of the first layer $R_1$ is smaller than the ZWXWL scheme. There are some other works on hierarchical caching model, such as the coded placement for hierarchical cache-enabled network \cite{TWGSZQX,TWGSSZQX}, the intension between two layers transmission loads \cite{LZX} and the topology where there are some users directly connected to the server \cite{TMHMMM} etc.


\subsection{Contribution and paper organization}
\label{subsec:con-org}

In this paper we focus on the hierarchical caching model in \cite{KNMD}. From the introduction in the above subsection, the existing schemes already have a good performance on $R_2$, while designing a scheme aiming at decreasing $R_1$ still remains open. Further, \cite{KNMD,ZWXWL,WWCY} all utilize the MN scheme or MN decentralized scheme, whose subpacketization increases exponentially with the growing on user number, so it is meaningful to design the scheme with minimizing the transmission load for the first layer $R_1$ or reducing the subpacketization. According to the above points, in this paper we obtain the following main results.
\begin{itemize}
\item   Inspired by the concept of PDA, we propose a new combination structure referred to as hierarchical placement delivery array (HPDA), which could be used to characterize the placement and delivery phase of a hierarchical coded caching scheme. So designing a hierarchical coded caching scheme is transformed to constructing an appropriate HPDA.
\item We propose a class of HPDA by dividing a MN PDA into several subarrays. Then we have a class of hierarchical coded caching schemes which achieve the lower bound of the first layer transmission load $R_1$.
\item We provide a general hybrid construction of HPDA based on any two PDAs. So we can get a scheme with flexible subpacketiztion when we choose the base PDAs with low subpacketiziations. In addition if the base PDAs are MN PDAs, then the scheme realized by our HPDA is exactly the scheme in \cite{WWCY}.
\end{itemize}
The rest of this paper is organized as follows. The hierarchical caching model and some preliminary results are introduced in Section \ref{sec:System}. We review PDA and introduce the structure of HPDA in Section \ref{sec:HPDA}. The main results and performance analysis are listed in Section \ref{sec:RESULT}. The proofs of our main results are proposed in Sections \ref{sec:5} and \ref{sec:6}. Finally we conclude this paper in Section \ref{sec:7}.

\subsection{Notations}
 The following notations are used in this paper.

\begin{itemize}
  \item For any positive integers $a$ and $b$ with $a<b$, let $[a:b]\triangleq\{a,a+1,\ldots,b\}$, $[a:b)\triangleq\{a,a+1,\ldots,b-1\}$ and $[a]\triangleq\{1,2,\ldots,a\}$. Let $\binom{[b]}{t}\triangleq \{\mathcal{V}|\mathcal{V}\subseteq [b],|\mathcal{V}|=t\}$, for any positive integer $t\leq b$.
  \item Given an array $\mathbf{P}=(p_{j,k})_{j\in[F], k\in[K]}$ with alphabet $[S]\cup \{*\}$, we define $\mathbf{P}+a=(p_{j,k}+a)_{j\in[F], k\in[K]}$ and $\mathbf{P}\times a=(p_{j,k}\times a)_{j\in[F], k\in[K]}$  for any integer $a$, where $a+*=*, a\times *=*$.
\end{itemize}

\section{Problem Definitions and Prior Works}
\label{sec:System}
     In this section, we first describe the hierarchical   caching system, and then review some   existing works   that motivate  our work in this paper.
\subsection{Hierarchical Caching System}
\label{subsec:hccs}
Consider  a $(K_1,K_2;M_1,M_2;N)$  hierarchical caching system as shown  Fig. \ref{fig-model}, which  consists of a single  server, $K_1$ mirror sites and $K_1K_2$ users. The server connects with  $K_1$ mirror sites via a shared link and each mirror site connects with $K_2$ users via another shared link. The server contains a collection of $N$ files, denoted by $\mathcal{W} = \{W_1, W_2,\ldots,W_{N}\}$, each of which is uniformly distributed over $[0,1]^B$. Each mirror site and user  has memory size of $M_1B$   and $M_2B$ bits, respectively, for some $M_1$, $M_2\geq 0$.

Denote the $k_2$-th user attached to the $k_1$-th mirror site as $\text{U}_{k_1,k_2}$, for  $k_1\in[K_1], k_2\in [K_2]$, and the set of users attached to the $k_1$-th mirror site as $\mathcal{U}_{k_1}$.  An $F$-division $(K_1,K_2;M_1,M_2;N)$ coded caching scheme contains two phases:
 \begin{itemize}
\item {\bf Placement phase:} During the off peak traffic time, each file is divided into $F$ packets with equal size, i.e., $W_{n}=\{W_{n,j}\ |\ n\in [N], j\in [F]\}$ where $B$ is divisible by $F$. Then the mirror sites and users cache some packets of each file. In other words, we consider the uncoded cache placement. Denote the contents cached by the   mirror site $k_1$ and user $\text{U}_{k_1,k_2}$ as $\mathcal{Z}_{k_1}$ and $\mathcal{Z}_{(k_1,k_2)}$, respectively. During the placement phase, we assume the server is not aware of the users' requests.

\item {\bf Delivery phase:} During the peak traffic time, each user requests one file from the file library $\mathcal{W}$ randomly. The  demand vector is denoted by $\mathbf{d} = (d_{1,1}, d_{1,2},\ldots,d_{k_1,k_2})$, i.e., user $\text{U}_{k_1,k_2}$,   requests the $d_{k_1,k_2}$-th file where $d_{k_1,k_2}\in [N]$. The messages sent in the hierarchical network contain two parts:

    \begin{itemize}
    \item \textbf{The messages sent by the server:} Based on the cached contents and the demand vector $\mathbf{d}$, the server broadcasts a message including ${S}({\bf d})$ packets   to $K_1$ mirror sites. 
    \item \textbf{The messages sent by mirror site:} Based on the messages sent by the server, the locally cached contents and the demand vector $\mathbf{d}$, each mirror site $k_1$ broadcasts a   coded messages of size $S_{k_1}({\bf d})$ packets to its attached users (i.e., users in $\mathcal{U}_{k_1}$), such that all the  users can recover their requested files.

    \end{itemize}
\end{itemize}

In this paper, we consider the \emph{worst} case where each user requests a distinct file. Given a hierarchical caching system described above, the transmission loads in terms of files for the first and second layer are defined as
\begin{eqnarray*}
&&R_1=  \max\left.\left\{\frac{S({\bf d})}{F}\ \right|\ {\bf d}\in [N]^{K_1K_2}\right\},\\
&&R_2=\max\left.\left\{ \frac{S_{k_1}({\bf d})}{F}\ \right|\ k_1\in[K_1],{\bf d}\in [N]^{K_1K_2}  \right \},
\end{eqnarray*}
respectively. Define the \emph{optimal} transmission loads of the first and second layer, denoted by  $R_1^*$ and $R_2^*$, as the minimum transmission load of $R_1$ and $R_2$ under uncoded placement respectively,  such that all users can recover their requesting files.

\subsection{Prior Works}
\label{subsec:prior-work}
The authors in \cite{KNMD} first studied this hierarchy   caching system and proposed a   decentralized hierarchical  coded caching scheme, namely the KNMD scheme,     based on the decentralized coded caching scheme for cache-aided broadcast network \cite{MND}.
The   main idea of KNMD scheme is to divide the system into two independent subsystems for some fixed parameters $\alpha$, $\beta$ $\in [0:1]$.  The first subsystem includes the entire cache memory of each mirror site and a $\beta$ fraction of each user's cache memory, which is responsible for caching and delivering the $\alpha$ parts of  each file.
The second subsystem includes the remaining $(1-\beta)$ fraction of each user's cache memory, and is responsible for caching and delivering the left $(1-\alpha)$ parts of each file.
In the first subsystem,  the server first sends coded signals to $K_1$ mirror sites using the single-layer MN decentralized coded caching scheme \cite{MND} where each mirror site  requests $K_2$ distinct files, without considering users'  cache contents. Then each mirror site   decodes its intend $K_2$ files and   applies again the single-layer decentralized coded caching scheme to broadcast a message to its attached users to satisfy their demands. In the second sub-system, the server  ignores the cache contents of mirror sites, and applies the   single-layer decentralized coded caching scheme to directly serve  $K_1K_2$ users each of caching size $(1-\beta) M_2$.    By extending this scheme to the case with centralized data placement \cite{MN}, we obtain     the transmission loads of the first and second layer, denoted by $R^\text{KNMD}_1$ and $R^\text{KNMD}_2$,   as
\begin{eqnarray}
\label{DHCC-B}
\begin{split}
R^\text{KNMD}_1(\alpha,\beta)&\triangleq   \alpha\cdot K_2\cdot r_c\left(\frac{M_1}{\alpha N},K_1\right)
+(1-\alpha)\cdot r_c\left(\frac{(1-\beta)M_2}{(1-\alpha)N},K_1K_2 \right),&\\
R^\text{KNMD}_2(\alpha,\beta)&\triangleq  \alpha\cdot r_c\left(\frac{\beta M_2}{\alpha N},K_2\right)
+(1-\alpha)\cdot r_c\left(\frac{(1-\beta)M_2}{(1-\alpha)N},K_2 \right),
\end{split}
\end{eqnarray} for some $\alpha$ and $\beta$, where
\begin{eqnarray*}
  r_c\left(\frac{M}{N},K\right)  \triangleq \frac{K (1-M/N)}{1+KM/N}
\end{eqnarray*}is the transmission load of the $(K,M,N)$ MN scheme for any memory ration $\frac{M}{N}\in\{0,\frac{1}{K},\frac{2}{K},\ldots,1\}$.



In the KNMD scheme, the server sends messages to the mirror sites while ignoring the users' cache contents. This means that the server may send some information which has already been stored by the users, leading to redundant communication cost in the first layer.  To address this problem,   \cite[Section VI]{WWCY}   improved the transmission load of the first layer of DHCC scheme by concatenating two MN schemes, whose transmission loads of the first and second layer, denoted by $R^\text{WWCY}_1$ and $R^\text{WWCY}_2$,  are 
\begin{eqnarray}
\label{eq-load-W}
\begin{split}
R^\text{WWCY}_1(\alpha,\beta)&\triangleq   \alpha \cdot r_c\left(\frac{M_1}{\alpha N},K_1\right)r_c\left(\frac{\beta M_2}{\alpha N},K_2\right)
+(1-\alpha)r_c\left(\frac{(1-\beta)M_2}{(1-\alpha)N},K_1K_2\right),\\
 R^\text{WWCY}_2(\alpha,\beta)&\triangleq  \alpha\cdot r_c\left(\frac{\beta M_2}{\alpha N},K_2\right)
+(1-\alpha)\cdot r_c\left(\frac{(1-\beta)M_2}{(1-\alpha)N},K_2 \right).
\end{split}
\end{eqnarray}


Note that under   uncoded placement, schemes in \cite{KNMD} and \cite[Section VI]{WWCY} both achieve the optimal transmission load of the second layer  when $\alpha=\beta$, i.e.,
$$R^*_2=r_c\left(\frac{  M_2}{  N},K_2\right).$$
An interesting question  is  what is the optimal  transmission load of the first layer. In this paper, we aim  to find novel centralized coded caching schemes  to reduce the transmission load of the first layer (i.e., $R_1$), and   establish its optimal value for some regimes.

\section{Hierarchy Placement Delivery Array}
\label{sec:HPDA}
In this section, we first briefly describe the vanilla PDA for the single-layer cache-aided broadcast network \cite{YCTC}, and then introduce a novel PDA structure, namely  HPDA, that would help characterize the placement and delivery of coded caching schemes for the hierarchical caching system in Fig. \ref{fig-model}.
\subsection{Placement Delivery Array}
\label{subsec:PDA}
\begin{definition}
\label{def-PDA}
(\cite{YCTC}) For any positive integers $K,F, Z$ and $S$, an $F\times K$ array $\mathbf{P}=(p_{j,k})_{j\in[F] ,k\in[K]}$ over alphabet set $\{*\}\bigcup [0,S)$ is called a $(K,F,Z,S)$ PDA if it satisfies the following conditions,
 \item [C$1$.] The symbol ``$*$" appears $Z$ times in each column;
 \item [C$2$.] Each integer occurs at least once in the array;
 \item [C$3$.] For any two distinct entries $p_{j_1,k_1}$ and $p_{j_2,k_2}$, $p_{j_1,k_1}=p_{j_2,k_2}=s$ is an integer only if
 \begin{enumerate}
 \item [a.] $j_1\ne j_2$, $k_1\ne k_2$, i.e., they lie in distinct rows and distinct columns; and
  \item [b.] $p_{j_1,k_2}=p_{j_2,k_1}=*$, i.e., the corresponding $2\times 2$ subarray formed by rows $j_1,j_2$ and columns $k_1,k_2$ must be of the following form
 \begin{align*}
 \left(\begin{array}{cc}
 s & *\\
 * & s
 \end{array}\right)~\textrm{or}~
 \left(\begin{array}{cc}
 * & s\\
 s & *
 \end{array}\right).
 \end{align*}
 \end{enumerate}
\end{definition}
\begin{example}
\label{example-1}
When $K=F=S=3$ and $Z=1$, we can see that the following array is a $(3,3,1,3)$ PDA.
\begin{eqnarray}
\label{eq-PDA()3313}
\mathbf{A}=\left(\small{
     \begin{array}{ccc}
       \ast & 1 & 2 \\
       1 & \ast & 3 \\
       2 & 3 & \ast \\
     \end{array}}
   \right).
   \end{eqnarray}
\end{example}

The authors in \cite{YCTC} showed that a $(K,F,Z,S)$ PDA can be used to realize an $F$-division $(K,M,N)$ coded caching scheme with $\frac{M}{N}=\frac{Z}{F}$ and transmission load $R=\frac{S}{F}$ for the single-layer cache-aided broadcast network. Furthermore the seminal coded cahcing scheme proposed in~\cite{MN} can be represented by a special PDA which is referred to as MN PDA. That is the following result.
\begin{lemma}\rm(MN PDA\cite{MN})
\label{le-MN}
For any positive integers $K$ and $t$ with $t\leq K$, there exists a   $\left(K,{K\choose t},{K-1\choose t-1},{K\choose t+1}\right)$ PDA which realizes a $(K,M,N)$ MN scheme with $\frac{M}{N}=\frac{t}{K}$, subpacketization $F={K\choose t}$ and transmission load $R=\frac{K-t}{t+1}$.
\hfill $\square$
\end{lemma}
Here we briefly review the construction of MN PDA as follows.
\begin{construction}\rm (MN PDA\cite{MN})
\label{con-MN} For any integer $t\in[K]$, let $F={K\choose t}$. Then we have a ${K\choose t}\times K$ array $\mathbf{P}=\left(\mathbf{P}(\mathcal{T},k)\right)_{\mathcal{T}\in {[K]\choose t}, k\in [K]}$ by
\begin{align}\label{Eqn_Def_AN}
\mathbf{P}(\mathcal{T},k)=\left\{\begin{array}{cc}
\phi_{t+1}(\mathcal{T}\cup\{k\}), & \mbox{if}~k\notin\mathcal{T}\\
*, & \mbox{otherwise},
\end{array}
\right.
\end{align}
where $\phi_{t+1}(\cdot)$ is a bijection from  $\binom{[K]}{t+1}$ to $[\binom{K}{t+1}]$ and the rows are labelled by all the subsets $\mathcal{T}\in {[K]\choose t}$ listed in the order from the small to large.
\hfill $\square$
\end{construction}

Finally we should point out that PDA has been widely studied. There are some schemes with lower subpacketization level based on PDA proposed in \cite{YCTC,CJYT,CJWY,CWZW,WCWG,MJW,ZCJ,ZCW,CWLZG,SSRS}. In addition, the authors in \cite{SKTADA} pointed out that all the proposed schemes in \cite{TR,SZG,STD,YTCC,KP} could be represented by appropriate PDAs.

\subsection{Hierarchical Placement Delivery Array}
The definition of hierarchy placement delivery array (HPDA) is given as follows.
\label{subsec:HPDA}
\begin{definition}\label{def-H-PDA}
For any given positive integers $K_{1}, K_{2}, F$, $Z_{1}$, $Z_{2}$ with $Z_1<F$, $Z_2<F$ and any integer sets $\mathcal{S}_\text{m}$ and $\mathcal{S}_{k_1}$, $k_1\in[K_1]$, an $F\times (K_1+K_1K_2)$ array $\mathbf{P}=(\mathbf{P}^{(0)},\mathbf{P}^{(1)},\ldots,\mathbf{P}^{(K_1)})$,
where $\mathbf{P}^{(0)}=(p^{(0)}_{j,k_1})_{j\in[F],k_1\in [K_1]}$ is an $F\times K_1$ array consisting of $*$ and null, and $\mathbf{P}^{(k_1)}=(p^{(k_1)}_{j,k_2})_{j\in[F],k_2\in [K_2]}$ is an $F\times K_2$ array over $\{*\}\bigcup \mathcal{S}_{k_1}$,  $k_1\in[K_1]$, is a $(K_1,K_2;F;Z_1,Z_2;\mathcal{S}_\text{m},\mathcal{S}_1,\ldots,\mathcal{S}_{K_1})$ hierarchy placement delivery array (HPDA) if it satisfies the following conditions:
\begin{itemize}
\item[B1.] Each column of $\mathbf{P}_0$ has $Z_1$ stars;
\item[B2.] $\mathbf{P}^{(k_1)}$ is a $(K_2,F,Z_2,|\mathcal{S}_{k_1}|)$ PDA for each $k_1\in [K_1]$.
\item[B3.] Each integer $s\in \mathcal{S}_\text{m}$ occurs in exactly one subarray $\mathbf{P}^{(k_1)}$ where $k_1\in[K_1]$. And for each $p^{(k_1)}_{j,k_2}=s\in \mathcal{S}_\text{m}$, $j\in[F],k_1\in[K_1],k_2\in[K_2]$, $p^{(0)}_{j,k_1}=*$ ;
\item[B4.] For any two entries $p^{(k_1)}_{j,k_2}$ and $p^{(k'_1)}_{j',k'_2}$ where $k_1\neq k'_1\in[K_1]$, $j,j'\in [F]$ and $k_2,k'_2\in[K_2]$, if $p^{(k_1)}_{j,k_2}=p^{(k'_1)}_{j',k'_2}$ is an integer then
\begin{itemize}
\item $p^{(k_1)}_{j',k_2}$ is an integer only if $p^{(0)}_{j',k_1}=*$;
\item $p^{(k'_1)}_{j,k'_2}$ is an iteger only if $p^{(0)}_{j,k'_1}=*$.
\end{itemize}
\end{itemize}
\end{definition}

\begin{small}
\begin{algorithm}[http!]
	\caption{Caching scheme based on $(K_1,K_2;F;Z_1,Z_2$; $\mathcal{S}_\text{m},\mathcal{S}_1,\ldots,\mathcal{S}_{K_1})$ HPDA $\mathbf{P}$}\label{alg:H-PDA}
	\begin{algorithmic}[1]
		\Procedure {Placement}{$\mathbf{P}$, $\mathcal{W}$}
		\State Split each file $W_n\in \mathcal{W}$ into $F$ packets, i.e., $W_{n}=\{W_{n,j}\ |\ j\in[F]\}$.
		\For{$k_1\in [K_1]$}
		\State $\mathcal{Z}_{k_1}\leftarrow\{W_{n,j}\ |\ p^{(0)}_{j,k_1}=*, n\in [N], j\in[F]\}$
		\EndFor
		\For{$(k_1,k_2), k_1\in[K_1],k_2\in[K_2]$}
		\State $\mathcal{Z}_{(k_1,k_2)}\leftarrow\{W_{n,j}\ |\ p^{(k_1)}_{j,k_2}=*, n\in [N], j\in[F]\}$
		\EndFor
		\EndProcedure
		\Procedure{Delivery\_Server}{$\mathbf{P}, \mathcal{W}, {\bf d}$}
		\For{$s\in \left(\bigcup_{k_1=1}^{K_1}\mathcal{S}_{k_1}\right)\setminus\mathcal{S}_\text{m} $}
		\State  Server sends the following coded signal to   the mirror sites: \\ \ \ \ \ \ \ \ \ \ \ \ \ $X_s=\bigoplus_{p^{(k_1)}_{j,k_2}=s,j\in[F],k_1\in[K_1],k_2\in[K_2]}W_{d_{k_1,k_2},j}$         \EndFor
		\EndProcedure
        \Procedure{Delivery\_Mirrors}{$\mathbf{P}, \mathcal{W}, {\bf d},X_s$}
        \For{$k_1\in [K_1], s\in\mathcal{S}_{k_1}\setminus\mathcal{S}_\text{m} $}
        \State After receiving $X_s$, mirror site $k_1$ sends the following coded signal  to users in $\mathcal{U}_{k_1}$:
		\State  $X_{k_1,s}=X_s\bigoplus\bigg(\!\!\!\!\!\!\bigoplus\limits_{
    \tiny\begin{array}{c}
    p^{(k'_1)}_{j,k_2}=s,p^{(0)}_{j,k'_1}=*,  k_2\in[K_2]\\
    j\in[F],k'_1\in [K_1]\backslash\{k_1\}
    \end{array}
    }\!\!\!\!\!\!\!\!W_{d_{k'_1,k_2},j}\bigg)$

		\EndFor
        \For{$k_1\in[K_1]$, $s'\in \mathcal{S}_{k_1}\bigcap\mathcal{S}_\text{m}$}
		\State Mirror site $k_1$  sends the following  coded signal to users in $\mathcal{U}_{k_1}$
		\State  \ \ \ \ \ \ \ \ $X_{k_1,s'}=\bigoplus_{p^{(k_1)}_{j,k_2}=s',j\in[F],k_2\in[K_2]}W_{d_{k_1,k_2},j}$
		\EndFor
		\EndProcedure
	\end{algorithmic}
\end{algorithm}
\end{small}
For any given HPDA, when we use $\mathbf{P}^{(0)}$ to indicate the data placement at mirror sites, and use $\mathbf{P}^{(k_1)}, k_1\in[K_1]$ to indicate data placement at the users attached to $k_1$-th mirror site (i.e., users in $\mathcal{U}_{k_1}$) and the delivery strategy at the server and mirror sites (see detailed explanation in Remark \ref{re-relationship-HPDA-HCCS}), a hierarchical coded caching scheme can be obtained by Algorithm \ref{alg:H-PDA}.

First we use the following example to demonstrate the placement and delivery strategy of the hierarchal caching scheme by Algorithm \ref{alg:H-PDA} based on a HPDA.
\begin{example}
\label{ex-1}
When $K_{1}=3$, $K_{2}=2$, $F=15$, $Z_{1}=$ $6$, $Z_{2}=4$, $\mathcal{S}_\text{m}=[7:42], \mathcal{S}_1=[1:18], \mathcal{S}_2=[1:6]\bigcup[19:30], \mathcal{S}_3=[1:6]\bigcup[31:42]$,  one can check that the following array is a $(3,2;15;6,4;\mathcal{S}_\text{m},\mathcal{S}_1,\mathcal{S}_2,\mathcal{S}_3)$ HPDA $\mathbf{P}$ in \eqref{eg-HPDA}.
\begin{eqnarray}
\label{eg-HPDA}
\begin{split}
\mathbf{P}&=(\mathbf{P_0}, \mathbf{P_1}, \mathbf{P_2}, \mathbf{P_3})\\
&=
\left(
\begin{array}{ccc|cc|cc|cc}
  * & * &   & 7  & 8  & 19 & 20 & 1  & 2 \cr
  * &   &   & 9  & 10 & *  & 1  & *  & 3 \cr
  * &   &   & 11 & 12 & *  & 2  & 3  & * \cr
  * &   &   & 13 & 14 & 1  & *  & *  & 4 \cr
  * &   &   & 15 & 16 & 2  & *  & 4  & * \cr
  * &   & * & 17 & 18 & 3  & 4  & 31 & 32\cr
    & * &   & *  & 1  & 21 & 22 & *  & 5 \cr
    & * &   & *  & 2  & 23 & 24 & 5  & * \cr
    &   & * & *  & 3  & *  & 5  & 33 & 34\cr
    &   & * & *  & 4  & 5  & *  & 35 & 36\cr
    & * &   & 1  & *  & 25 & 26 & *  & 6 \cr
    & * &   & 2  & *  & 27 & 28 & 6  & * \cr
    &   & * & 3  & *  & *  & 6  & 37 & 38\cr
    &   & * & 4  & *  & 6  & *  & 39 & 40\cr
    & * & * & 5  & 6  & 29 & 30 & 41 & 42
\end{array}
\right).
\end{split}
\end{eqnarray}

Based on $\mathbf{P}$ and by Algorithm \ref{alg:H-PDA}, we can get a $15$-$(3,2;2.4,1.6;6)$ coded caching scheme in the following way.

\begin{itemize}
	\item \textbf{Placement Phase}: From Line 2 in Algorithm \ref{alg:H-PDA}, each file is divided into $F=15$ packets with equal size, i.e., $W_{n}=\{W_{n,1}, W_{n,2}, \ldots, W_{n,15}\}, n\in [6]$. From lines 3-5 in Algorithm \ref{alg:H-PDA} and $\mathbf{P}_0$ in \eqref{eg-HPDA}, the contents cached by mirror sites are as follows:
\begin{eqnarray*}
&&\mathcal{Z}_{1} =\{W_{n,1}, W_{n,2}, W_{n,3}, W_{n,4}, W_{n,5}, W_{n,6}\ |\ n\in[6]\},\\
&&\mathcal{Z}_{2} =\{W_{n,1}, W_{n,7}, W_{n,8}, W_{n,11}, W_{n,12}, W_{n,15}\ |\ n\in[6]\},\\
&&\mathcal{Z}_{3} =\{W_{n,6}, W_{n,9}, W_{n,10}, W_{n,13}, W_{n,14}, W_{n,15}\ |\ n\in[6]\}.
\end{eqnarray*}From lines 6-8 in Algorithm \ref{alg:H-PDA} and $\mathbf{P}_{k_1}$ in \eqref{eg-HPDA}, $k_1\in[3]$, the packets cached by the users are as follows:
\begin{eqnarray*}
&&\mathcal{Z}_{(1,1)} =\{W_{n,7}, W_{n,8}, W_{n,9}, W_{n,10}\ |\  n\in[6]\},\\
&&\mathcal{Z}_{(1,2)} =\{W_{n,11}, W_{n,12}, W_{n,13}, W_{n,14}\ |\  n\in[6]\},\\
&&\mathcal{Z}_{(2,1)} =\{W_{n,2}, W_{n,3}, W_{n,9}, W_{n,13}\ |\  n\in[6]\},\\
&&\mathcal{Z}_{(2,2)} =\{W_{n,4}, W_{n,5}, W_{n,10}, W_{n,14}\ |\  n\in[6]\},\\
&&\mathcal{Z}_{(3,1)} =\{W_{n,2}, W_{n,4}, W_{n,7}, W_{n,11}\ |\  n\in[6]\},\\
&&\mathcal{Z}_{(3,2)} =\{W_{n,3}, W_{n,5}, W_{n,8}, W_{n,12}\ |\  n\in[6]\}.
\end{eqnarray*}	
\item\textbf{Delivery Phase}: Assume that ${\bf d}=(1,2,3,4,5,6)$. From Algorithm \ref{alg:H-PDA}, the messages sent to all users consist of two parts.
\begin{itemize}
\item The messages $S({\bf d})$ sent by the server: From \eqref{eg-HPDA} and Line 11, we have $s\in\left(\bigcup_{k_1=1}^{3}\mathcal{S}_{k_1}\right)\setminus\mathcal{S}_\text{m}=[1:6]$. By Lines 11-14 in Algorithm \ref{alg:H-PDA}, the sever transmits the coded messages $X_{s}$:
\begin{equation*}
	\begin{split}
        &W_{1,11}\oplus W_{2,7}\oplus W_{3,4}\oplus W_{4,2}\oplus W_{5,1}, \\
        &W_{1,12}\oplus W_{2,8}\oplus W_{3,5}\oplus W_{4,3}\oplus W_{6,1}, \\
        &W_{1,13}\oplus W_{2,9}\oplus W_{3,6}\oplus W_{5,3}\oplus W_{6,2}, \\
        &W_{1,14}\oplus W_{2,10}\oplus W_{4,6}\oplus W_{5,5}\oplus W_{6,4}, \\
        &W_{1,15}\oplus W_{3,10}\oplus W_{4,9}\oplus W_{5,8}\oplus W_{6,7}, \\
        &W_{2,15}\oplus W_{3,14}\oplus W_{4,13}\oplus W_{5,12}\oplus W_{6,11},
	\end{split}
\end{equation*}	to all mirror sites. So the transmission load of the first layer is $R_1=\frac{6}{15}=0.4$.
\item The messages $S_{k_1}({\bf d})$ sent by mirror site $k_1$ consists of the coded packets $X_{k_1,s}$ generated by $X_s$ from the server and the packets cached by mirror site $k_1$ where $s\in\mathcal{S}_{k_1}\setminus\mathcal{S}_\text{m}$, and the coded packets $X_{k_1,s'}$ generated only by the packets cached by mirror site $k_1$ where $s'\in \mathcal{S}_{k_1}\bigcap \mathcal{S}_{\text{m}}$.

From Lines 17-20 and \eqref{eg-HPDA}, we have $\mathcal{S}_{1}\setminus\mathcal{S}_\text{m}=[1:6]$, so the mirror site 1 sends the coded packets $X_{1,s}$:
\begin{IEEEeqnarray}{rCl}
&&~\left(W_{1,11}\oplus W_{2,7}\oplus W_{3,4}\oplus W_{4,2}\oplus W_{5,1}\right)  \oplus\left(W_{3,4}\oplus W_{4,2}\oplus W_{5,1}\right)\nonumber
=~W_{1,11}\oplus W_{2,7}\nonumber,\\
&&~\left(W_{1,12}\oplus W_{2,8}\oplus W_{3,5}\oplus W_{4,3}\oplus W_{6,1}\right)  \oplus\left(W_{3,5}\oplus W_{4,3}\oplus W_{6,1}\right)\nonumber
=~W_{1,12}\oplus W_{2,8}\nonumber,\\
&&~\left(W_{1,13}\oplus W_{2,9}\oplus W_{3,6}\oplus W_{5,3}\oplus W_{6,2}\right)  \oplus\left(W_{3,6}\oplus W_{5,3}\oplus W_{6,2}\right)\nonumber
=~W_{1,13}\oplus W_{2,9}\nonumber,\\
&&~\left(W_{1,14}\oplus W_{2,10}\oplus W_{4,6}\oplus W_{5,5}\oplus W_{6,4}\right)  \oplus\left(W_{4,6}\oplus W_{5,5}\oplus W_{6,4}\right)\nonumber
=~W_{1,14}\oplus W_{2,10}\nonumber,\\
&&~W_{1,15}\oplus W_{3,10}\oplus W_{4,9}\oplus W_{5,8}\oplus W_{6,7}\nonumber, \\
&&~W_{2,15}\oplus W_{3,14}\oplus W_{4,13}\oplus W_{5,12}\oplus W_{6,11}\nonumber,
\end{IEEEeqnarray}
since it can receive the coded packets $X_{s}$ from server and it has cached packets $\{W_{n,1}$, $W_{n,2}$, $W_{n,3}$, $W_{n,4}$, $W_{n,5}$, $W_{n,6}\ |\ n\in[6]\}$.
Then user $\text{U}_{1,1}$ can decode $W_{1,11}$, $W_{1,12}$, $W_{1,13}$, $W_{1,14}$ and $W_{1,15}$ from $X_{1,s}$ since it has cached $\{W_{n,7}, W_{n,8}, W_{n,9}, W_{n,10}\ |\  n\in[6]\}$. Similarly, $\text{U}_{1,2}$ can also recover some of its required file packets from $X_{1,s}$ and its own cache memory respectively. Since $|\mathcal{S}_1\setminus \mathcal{S}_{\text{m}}|=6$ there are $6$ coded packets sent by mirror site $1$.

Now we see the coded packets $X_{k_1,s'}$ sent by mirror site $k_1$. From Lines 21-24 and \eqref{eg-HPDA}, we have $\mathcal{S}_1\bigcap \mathcal{S}_{\text{m}}=[7:18]$, so mirror site $1$ sends $X_{1,s'}$:
    \begin{eqnarray*}
      &W_{1,1}, W_{1,2}, W_{1,3}, W_{1,4}, W_{1,5}, W_{1,6}, W_{2,1}, W_{2,2},\\
      &W_{2,3}, W_{2,4}, W_{2,5}, W_{2,6},
    \end{eqnarray*}
    to users $\text{U}_{1,1}$ and $\text{U}_{1,2}$ from its own cached packets. Clearly each user can directly get the above packets and there are $12$ packets. Then the transmission amount by mirror site $1$ is $\frac{6+12}{15}=1.2$, which is the transmission load of the second layer $R_2=1.2$.
\end{itemize}
\end{itemize}
Actually, $\mathbf{P}$ in \eqref{eg-HPDA} is obtained by Theorem \ref{th-maint-2}, whose $R_1$ achieves the minimum load under the restriction of parameters specified by $\mathbf{P}$. In \cite{KNMD,WWCY}, by the exhaustive computer searches for the values of $\alpha$ and $\beta$ to find the minimum transmission load of the first layer under the same circumstance, we have $R^{\text{KNMD}}_1=0.73$ from \eqref{DHCC-B} and $R^{\text{WWCY}}_1=0.55$ from \eqref{eq-load-W}. Clearly $R_1<R^{\text{WWCY}}_1<R^{\text{KNMD}}_1$.
\end{example}

\begin{remark}
\label{re-relationship-HPDA-HCCS}
From Algorithm \ref{alg:H-PDA} and Example \ref{ex-1}, we have the following relationship between $(K_1,K_2;F;Z_1,Z_2;\mathcal{S}_\text{m}$, $\mathcal{S}_1,\ldots,\mathcal{S}_{K_1})$ HPDA and its realized $F$-division coded caching scheme for the $(K_1,K_2; M_1,M_2; N)$ hierarchical coded caching problem where $\frac{M_1}{N}=\frac{Z_1}{F}$, $\frac{M_2}{N}=\frac{Z_2}{F}$.
\begin{itemize}
\item An $F\times K_1$ mirror sites-placement array $\mathbf{P}^{(0)}$ consists of $*$ and null entries. The column labels represent the mirror site indices while the row labels represent the packet indices. If entry $p^{(0)}_{j,k_1}= *$, $j\in[F]$ and $k_1 \in [K_1]$, then mirror site $k_1$ has already cached the $j$-th packet of all the files in server. All mirror sites have the same memory ratio $\frac{M_1}{N}=\frac{Z_1}{F}$ according to B1 of Definition \ref{def-H-PDA}.
\item An $F\times K_1K_2$ users-placement and delivery array $(\mathbf{P}^{(1)},\ldots,\mathbf{P}^{(K_1)})$ consists of $ \{$*$\}\bigcup\{\bigcup_{k_1=1}^{K_1}\mathcal{S}_{k_1}\}$. The column labels represent the user indices while the row labels represent the packet indices. If entry $p^{(k_1)}_{j,k_2} =*$, $j\in[F]$, $k_1\in[K_1]$, $k_2 \in [K_2]$, user $\text{U}_{k_1,k_2}$ has already cached the $j$-th packet of all the files in server. All the users have the same memory ratio $\frac{M_2}{N}=\frac{Z_2}{F}$ according to B$2$ of Definition \ref{def-H-PDA}. The integers in $\left(\bigcup_{k_1=1}^{K_1}\mathcal{S}_{k_1}\right)\setminus\mathcal{S}_\text{m}$ indicate the broadcast packets transmitted by the server, and the integers in $\mathcal{S}_{k_1}$, $k_1\in[K_1]$, represent the broadcast packets sent by the mirror site $k_1$. In addition the integers in $\mathcal{S}_\text{m}$ represent the multicast messages sent only by the mirror sites.
\item The property B2 and B4 of Definition \ref{def-H-PDA} guarantee that each user $\text{U}_{k_1,k_2}$ can recover its requested packet, since user $\text{U}_{k_1,k_2}$ or mirror site $k_1$ has cached all the other packets in the broadcast message except the one requested by $\text{U}_{k_1,k_2}$. More precisely, if entry $p^{(k_1)}_{j,k_2} =s\in\left(\bigcup_{k_1=1}^{K_1}\mathcal{S}_{k_1}\right)\setminus\mathcal{S}_\text{m}$, $j\in[F]$, $k_1\in[K_1]$, $k_2 \in [K_2]$, then the $j$-th packet of all files is not stored by user $\text{U}_{k_1,k_2}$. In this case the server broadcasts a coded packet (i.e., the XOR of all the requestd packets indicated by $s$) to the mirror sites. Assume that the packet required by user $\text{U}_{k_1,k_2}$, say $W_{d_{k_1,k_2},j}$, and any other packet, say $W_{d_{k'_1,k'_2},j'}$, are included in the coded signal $X_s$ listed in Line 13 of Algorithm \ref{alg:H-PDA}. Then we have $p^{(k_1)}_{j,k_2}=p^{(k'_1)}_{j',k'_2}=s$. If $k_1=k_1'$, then from the Condition C$3$ of definition \ref{def-PDA} we have $p^{(k_1)}_{j',k_2}=p^{(k_1')}_{j,k_2'}=*$ because $\mathbf{P}^{(k_1)}$ is a PDA, which means user $\text{U}_{k_1,k_2}$ has cached the packet $W_{d_{k'_1,k'_2},j'}$. If $k_1\neq k_1'$ and $p^{(k_1)}_{j',k_2}\neq *$, then from Condition B$4$ of Definition \ref{def-H-PDA} we have $p^{(0)}_{j',k_1}=*$. This implies that mirror site $k_1$ has cached $W_{d_{k'_1,k'_2},j'}$. From Line 19 of Algorithm \ref{alg:H-PDA}, the coded signal $X_{k_1,s}$, which is transmitted to user $\text{U}_{k_1,k_2}$ by mirror site $k_1$, is generated by cancelling the $W_{d_{k'_1,k'_2},j'}$ by mirror site $k_1$. So $X_{k_1,s}$ only contains one packet required by user $\text{U}_{k_1,k_2}$ and the packets which have been cached by user $\text{U}_{k_1,k_2}$. Clearly user $\text{U}_{k_1,k_2}$ can decode its requiring packet $W_{d_{k_1,k_2},j}$. So the number of packets transmitted by the server is $|\bigcup_{k_1=1}^{K_1}\mathcal{S}_{k_1}|-|\mathcal{S}_\text{m}|$. Then the transmission load from server to mirror sites is $R_1=\frac{|\bigcup_{k_1=1}^{K_1}\mathcal{S}_{k_1}|-|\mathcal{S}_\text{m}|}{F}$. While if $s'\in \mathcal{S}_\text{m}\bigcap\mathcal{S}_{k_1}$, by the Condition B$3$ of Definition \ref{def-H-PDA} the mirror site $k_1$ has already cached all the required packets labeled by $s'$. So the mirror site can broadcast a multicast message $X_{k_1,s'}$ (i.e. the XOR of all the requested packets indicated by $s'$) to the user in $\mathcal{U}_{k_1}$. Then the number of packets transmitted simply by the mirror site $k_1$ is $|\mathcal{S}_{k_1}|$. This implies that the transmission load from mirror site $k_1$ to its attached users in $\mathcal{U}_{k_1}$ is $R_2=\max_{k_1\in[K_1]}\left\{\  \frac{\mid\mathcal{S}_{k_1}\mid}{F}\ \right\}$.

\end{itemize}
\end{remark}

From the above investigations in Remark \ref{re-relationship-HPDA-HCCS}, we can obtain the following result.
\begin{theorem}
\label{th-main-result}
Given a $(K_1,K_2;F;Z_1,Z_2;\mathcal{S}_\text{m},\mathcal{S}_1,\ldots,\mathcal{S}_{K_1}) $ HPDA $\mathbf{P}=(\mathbf{P}_0,\mathbf{P}_1,\ldots,\mathbf{P}_{K_1})$,   we can obtain an $F$-division $(K_1,K_2;M_1,M_2;N)$ coded caching scheme with $\frac{M_1}{N}=\frac{Z_1}{F}$, $\frac{M_2}{N}=\frac{Z_2}{F}$ and transmission load $R_1=\frac{|\bigcup_{k_1=1}^{K_1}\mathcal{S}_{k_1}|-|\mathcal{S}_\text{m}|}{F}$, $R_2=\max_{k_1\in[K_1]}\left\{\  \frac{\mid\mathcal{S}_{k_1}\mid}{F}\ \right\}$.
\end{theorem}
From Theorem \ref{th-main-result}, we can obtain a hierarchical coded caching scheme by constructing an appropriate HPDA. So in this paper we focus on constructing HPDA to get its realized scheme with better performance compared with the previously known results.

\section{Main Results}
\label{sec:RESULT}

In this section, we first present new upper bounds on the optimal transmission loads $(R^*_1,R^*_2)$ based on two classes of HPDAs, and then compare these bounds with previously known results.
\begin{theorem}
\label{th-maint-2}
    For any positive integers $K_1$, $K_2$, $t$, $K_2<t<K_1K_2$, there exists   a $(K_1,K_2$; ${K_1K_2\choose t}$; ${K_1K_2-K_2\choose t-K_2}$, ${K_1K_2-1\choose t-1}-{K_1K_2-K_2\choose t-K_2}$; $\mathcal{S}_\text{m},\mathcal{S}_1$, $\ldots$, $\mathcal{S}_{K_1})$ HPDA,  which leads to an $F$-division $(K_1$, $K_2$; $M_1$, $M_2$; $N)$ coded caching scheme with
\begin{subequations}
\begin{IEEEeqnarray}{rCl}
  \label{th-para}
      &&\text{memory ratios}:  \frac{M_1}{N}=\frac{{K_1K_2-K_2\choose t-K_2}}{{K_1K_2\choose t}} ,\nonumber \\
      && \ \ \ \ \ \ \ \ \ \ \ \ \ \ \ \ \ \ \ \frac{M_2}{N}=\frac{t}{K_1K_2}-\frac{{K_1K_2-K_2\choose t-K_2}}{{K_1K_2\choose t}},~\quad\label{eqRatioThm2} \\[0.2cm]
      &&\text{subpacketization}:  F={K_1K_2\choose t},\label{eqPackThm2}\\[0.2cm]
      &&\text{transmission
    loads}:  R_1=\frac{K_1K_2-t}{t+1},\label{eqR1Thm2}\\ &&\hspace{1ex}R_2=\frac{K_1K_2-t}{t+1}-\frac{{K_1K_2-K_2\choose t+1}}{{K_1K_2\choose t}}+\frac{{K_1K_2-K_2\choose t-K_2}K_2}{{K_1K_2\choose t}}.\label{eqR2Thm2}
\end{IEEEeqnarray}
\end{subequations}
\end{theorem}
\begin{proof}
See the proof in Section \ref{sec:5}.
\end{proof}

\begin{corollary}
\label{remark-optimal}
For a two-level hybrid   network with  {memory ratios} satisfying \eqref{eqRatioThm2}, the transmission load $R_1$ in \eqref{eqR1Thm2} is  \emph{optimal}  under the uncoded data placement, i.e.,
$$ R_1^*=\frac{K_1K_2(1-\frac{M_1+M_2}{N})}{K_1K_2\frac{M_1+M_2}{N}+1}.$$
\begin{proof}
The achievability proof holds directly from Theorem \ref{th-maint-2}. For the converse proof, please refer to  Appendix \ref{appendix-optimal}.
\end{proof}
  It can be checked that the HPDA  in  \eqref{eg-HPDA} in  Example \ref{ex-1} is in fact a specific HPDA of Theorem \ref{th-maint-2}. In Section \ref{sub-sketch-group}, we show how to construct  the HPDA   in \eqref{eg-HPDA} based on a $(6,15,10,6)$ MN PDA.

\end{corollary}

The memory ratios in Theorem  \ref{th-maint-2} are constrained by   combination numbers  as shown in \eqref{eqRatioThm2}, which means the rate $R_1={K_1K_2-t}/({t+1})$ may be not always   achievable  for  general memory ratios  $(M_1/N,M_2/N)$. In order to allow flexible   memory ratios, we propose the following upper bound based on a new class of HPDA.    
\begin{theorem}
\label{th-maint-4}
For any $(K_1,F_1,Z_1,S_1)$ PDA $\mathbf{A}$ and $(K_2,F_2,Z_2,S_2)$ PDA $\mathbf{B}$, there exists a $(K_1,K_2$; $F_1F_2$; $Z_1F_2$, $Z_2F_1$; $\mathcal{S}_\text{m}$, $\mathcal{S}_1$, $\ldots$, $\mathcal{S}_{K_1})$ HPDA, which leads to a $(K_1,K_2;M_1,M_2;N)$ coded caching scheme with memory ratios $\frac{M_1}{N}=\frac{Z_1}{F_1}$,  $\frac{M_2}{N}=\frac{Z_2}{F_2}$ and transmission loads
\begin{eqnarray}
\label{co2-para}
R_1=\frac{S_1S_2}{F_1F_2},~ R_2=\frac{S_2}{F_2}.
\end{eqnarray}
\end{theorem}
\begin{proof}
See the proof in Section \ref{sec:6}
\end{proof}

\begin{remark}\label{remark-3} By choosing different PDAs $\mathbf{A}$ and $\mathbf{B}$ to construct the HPDA in Theorem \ref{th-maint-4}, we can obtain different transmission loads and subpacketization levels. In particular,
\begin{itemize}
\item
when $\mathbf{A}$ and $\mathbf{B}$ both are   MN PDA, the corresponding scheme is the same as the WWCY scheme \cite{WWCY}, which achieves the optimal transmission load of $R^*_2$;
\item when    $\mathbf{A}$ and $\mathbf{B}$  are the PDA proposed in \cite{YCTC} and   MN PDA,  respectively,  we could reduce the subpacketization level at cost of   increasing communication loads. We name the corresponding scheme as Scheme I for Theorem \ref{th-maint-4};
\item when  both $\mathbf{A}$ and $\mathbf{B}$  are   PDAs proposed in \cite{YCTC},  we would further reduce the  subpacketization level. We name the corresponding scheme as Scheme II for Theorem \ref{th-maint-4}.
\end{itemize}
\end{remark}


In Fig. \ref{per-analisis}, we compare the following schemes: 1) the KNMD scheme  \cite{KNMD}; 2) the  WWCY scheme  \cite{WWCY}; 3) the Scheme for Theorem \ref{th-maint-2};  4) Scheme I for Theorem \ref{th-maint-4}; 5)  Scheme II for Theorem \ref{th-maint-4}.  Note that we can also design our new hybrid schemes  like  the previous works \cite{KNMD,WWCY}, which  divide the system into two subsystems  with splitting parameters ($\alpha, \beta$), and run  the proposed schemes in the first subsystem, and the MN  scheme in the second subsystem. Since the optimal $\alpha$ and $\beta$ are hard to determine due to a tradeoff between $R_1$ and $R_2$, and  the second subsystem  totally ignores mirror sites'  caching abilities, we only focus on    schemes working in the first subsystem, i.e., compare all schemes with  $\alpha=\beta=1$. Besides, due to  the  limitation on memory ratios \eqref{eqRatioThm2}, it is hard to compare all schemes with   general $M_1,M_2\in[0:N]$.  We thus evaluate the performance of various scheme with fixed parameters $(K_1,K_2,N)=(40,20,800)$, and varying
parameters $(M_1,M_2)$ such that the ratios in \eqref{eqRatioThm2} are satisfied.
 More precisely, $M_1/N$ takes the value from $0.2$ to $0.9$ regularly with step size $0.1$, and $M_2/N$ takes the value from $0.72$ to $0.1$ (without fixed step size but on a downward trend), which satisfies  memory ratios \eqref{eqRatioThm2} in Theorem \ref{th-maint-2}.

\begin{figure}
\centering
\subfigure[]
{
\includegraphics[width=5.5cm]{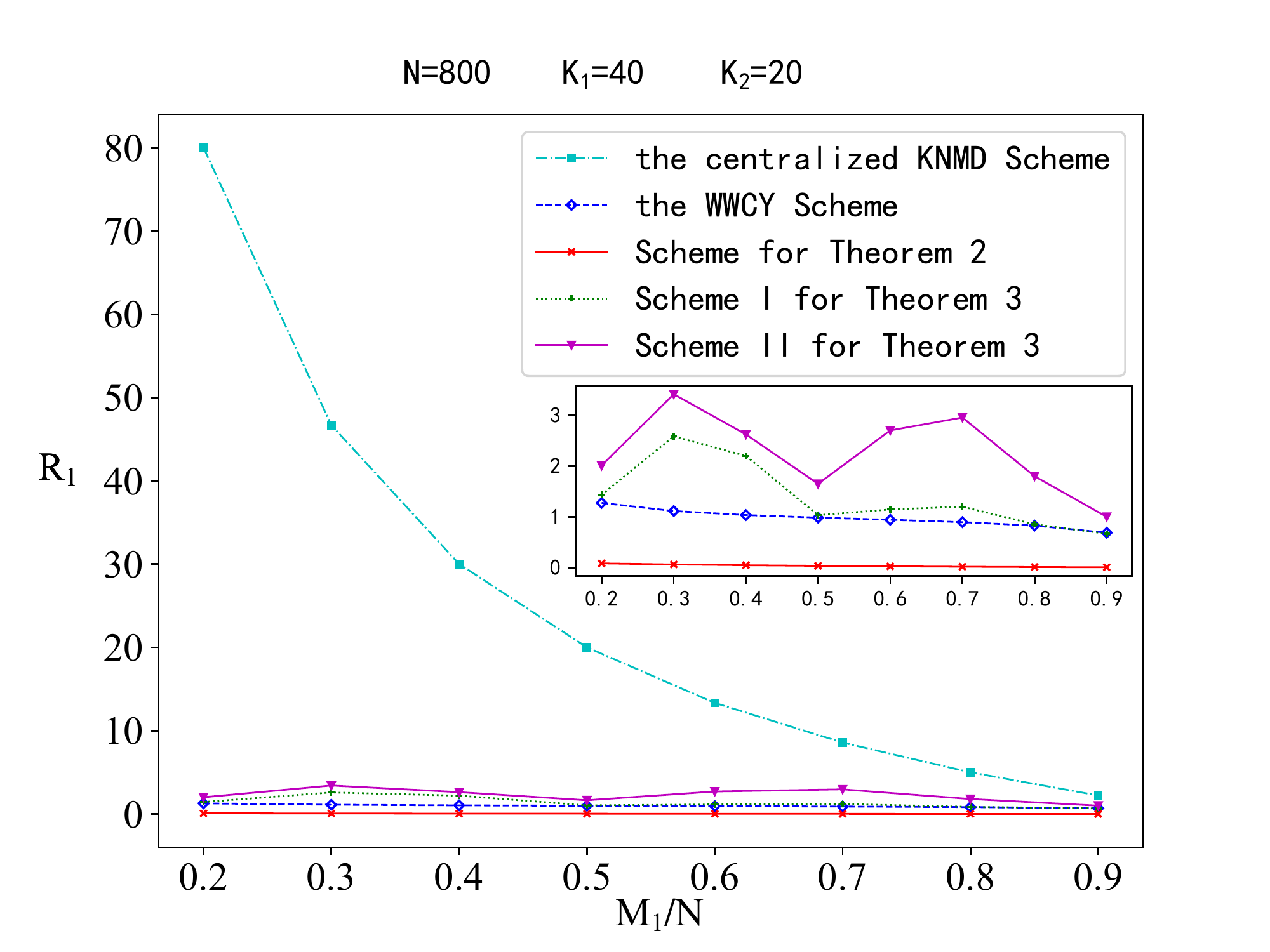}
\label{subfig1}
}
\subfigure[]{
\includegraphics[width=5.5cm]{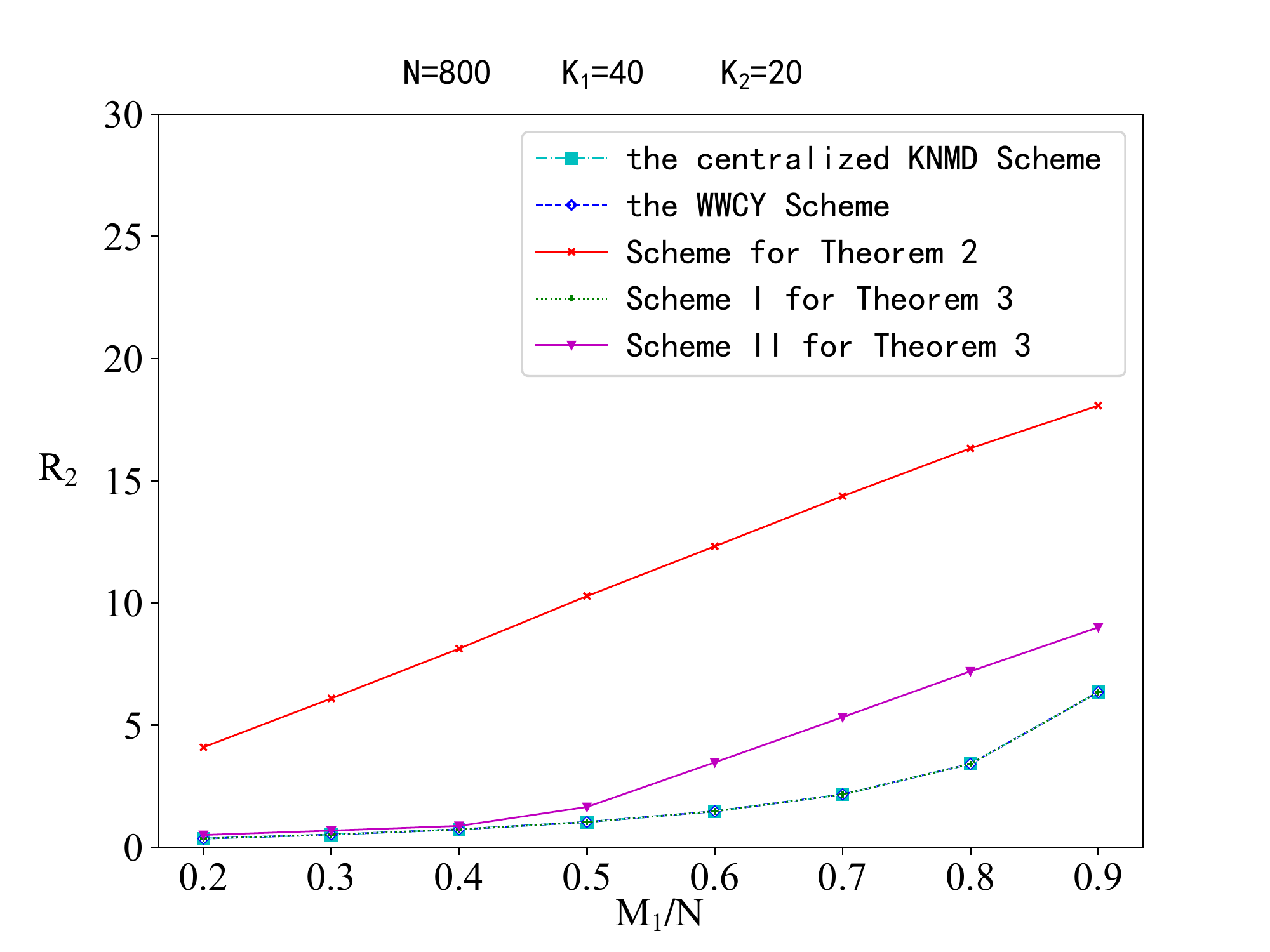}
\label{subfig2}
}
\subfigure[]{
\includegraphics[width=5.5cm]{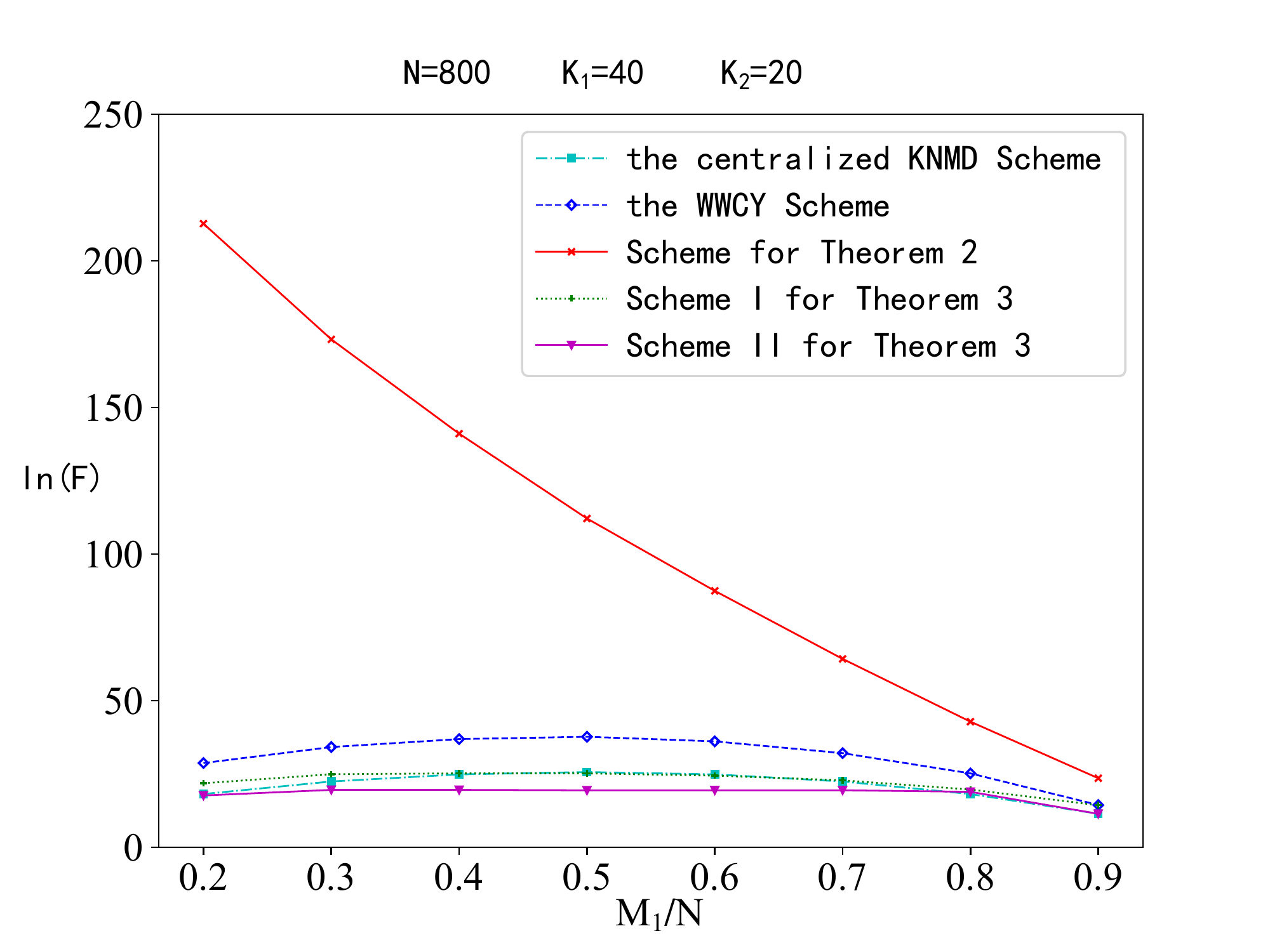}
\label{subfig3}
}
\caption{Performance comparison for a  $(K_1,K_2;M_1,M_2;N)$ caching system with $N=800$, $K_1=40$, $K_2=20$.}
\label{per-analisis}
\end{figure}



Fig. \ref{subfig1} plots the rate $R_1$ versus  ${M_1}/{N}$. We can see that applying the proposed   schemes for Theorem \ref{th-maint-2} and \ref{th-maint-4}  can significantly reduce the transmission load $R_1$  compared to the KNMD scheme. In order to have a clear view, we draw a sketch sub-figure in Fig. \ref{subfig1} that is without the KNMD scheme. Among all schemes, the scheme for Theorem \ref{th-maint-2} achieves the smallest $R_1$, and  the WWCY scheme achieves the second best performance. Note that the WWCY scheme is a special case of Theorem \ref{th-maint-4} where both $\mathbf{A}$ and $\mathbf{B}$ are MN PDAs, as mentioned in Remark \ref{remark-3}. By comparing the WWCY scheme, Scheme I and Scheme II for Theorem \ref{th-maint-4}, we can see that using different PDAs for $\mathbf{A}$ and $\mathbf{B}$ results in different transmission loads,  and using MN PDA for   $\mathbf{A}$ or $\mathbf{B}$ would reduce the transmission load than using other PDAs. 

Fig. \ref{subfig2} compares the rate $R_2$ versus ${M_1}/{N}$. It can be seen that the scheme for Theorem \ref{th-maint-2} requires the largest  $R_2$. In view of  Fig. \ref{subfig1} where  the scheme for Theorem \ref{th-maint-2} achieves the optimal $R^*_1$, we obtain that there exists a tradeoff between $R_1$ and $R_2$. In other words, minimizing  $R_1$ may lead to the increasing on $R_2$.  Note that the WWCY scheme, the KNMD scheme and the Scheme I for Theorem \ref{th-maint-4} achieve the same $R_2$, which is the   optimal $R^*_2$ under uncoded placement,  and the gap between the optimal rate $R^*_2$ and the rate $R_2$ of the Scheme II for Theorem \ref{th-maint-4} is almost marginal, especially when $M_1/N$ is small.  Note that the curves in Fig.  \ref{subfig2}   show  that $R_2$ increases with   ${M_1}/{N}$. This is because $R_2$ in general decreases with ${M_2}/{N}$, while ${M_2}/{N}$ increases due to   the relation $M_1/N+M_2/N=t/K$ indicated by  \eqref{eqRatioThm2}. 

Fig. \ref{subfig3} demonstrates the subpacketization levels of various schemes. It can be seen that the Scheme for Theorem  \ref{th-maint-2}, which achieves the minimum $R_1$, requires the highest subpacketization level, and the subpacketization level decreasing almost linearly with $M_1/N$. The Scheme II for Theorem \ref{th-maint-4}, which uses the proposed PDAs in \cite{YCTC} for both $\mathbf{A}$ and $\mathbf{B}$, requires the lowest subpacketization level. The WWCY scheme, which uses MN-type PDAs for  both $\mathbf{A}$ and $\mathbf{B}$, incurs larger subpacketization level than other schemes not using MN-type PDAs. From the above, we can conclude that by choosing different types of PDAs to construct the HPDA, one can achieve a flexible tradeoff between  the subpacketization level and transmission loads.

\section{Proof of Theorem \ref{th-maint-2}}
\label{sec:5}
 {  In this section, we   describe how to    construct a $(K_1,K_2$; $F$; $Z_1$, $Z_2$; $\mathcal{S}_\text{m},\mathcal{S}_1$, $\ldots$, $\mathcal{S}_{K_1})$ HPDA $\mathbf{P}= \left(\mathbf{P}^{(0)}\right.$, $\mathbf{P}^{(1)}$, $\ldots$, $\left.\mathbf{P}^{(K_1)}\right)$ in Theorem \ref{th-maint-2}  based on the array $(K,F,Z,S)$ MN PDA $\mathbf{Q}$. We partition $\mathbf{Q}$ into $K_1$ parts by column, i.e., $\mathbf{Q}=\left(\mathbf{Q}^{(1)},\ldots,\mathbf{Q}^{(K_1)}\right)$. Then we have the following expressions.
 \begin{IEEEeqnarray}{rCl}
&& \mathbf{P}^{(0)}=(p^{(0)}_{f,k_1})_{f\in[F],k_1\in [K_1]},~p^{(0)}_{f,k_1}\in\{*,null\}\\
&& \mathbf{P}^{(k_1)}=(p^{(k_1)}_{f,k_2})_{f\in[F],k_2\in [K_2]}, ~p^{(k_1)}_{f,k_2}\in\{*\}\cup \mathcal{S}_{k_1},~k_1\in [K_1]~\\
&&\mathbf{Q}^{(k_1)}=(q^{(k_1)}_{f,k_2})_{f\in[F],k_2\in [K_2]}, ~q^{(k_1)}_{f,k_2}\in[S], ~k_1\in [K_1].
\end{IEEEeqnarray}
As mentioned in Remark \ref{re-relationship-HPDA-HCCS},  $\mathbf{P}^{(0)}$ and $\mathbf{P}^{(k_1)}$  indicates the data placement   at mirror sites and users in $\mathcal{U}_{k_1}$, respectively.
Given any array,  we call a row of it \emph{star row} if this row   contains only star entries.   The construction of HPDA $\mathbf{P}= \left(\mathbf{P}^{(0)}\right.$, $\mathbf{P}^{(1)}$, $\ldots$, $\left.\mathbf{P}^{(K_1)}\right)$ for Theorem \ref{th-maint-2} can be briefly described as follows: To construct array $\mathbf{P}^{(0)}$, we let  its element  $p^{(0)}_{f,k_1}$    be a star entry   if the $f$-th row of $\mathbf{Q}^{(k_1)}$ is a star row, and be   null otherwise. To construct the array   $\mathbf{P}^{(k_1)}$, we simply replace all the star entries in each star row of   $\mathbf{Q}^{(k_1)}$ with distinct integers,  for all $k_1\in[K_1]$.

In the follows,  we first use an illustrative example to show the construction of $\mathbf{P}$, and then present our general proof of coded caching scheme based on the HPDA in Theorem \ref{th-maint-2}.
\subsection{Example of the Construction of HPDA in Theorem \ref{th-maint-2}}
\label{sub-sketch-group}

\begin{figure}[http!]
\centering
\includegraphics[scale=0.9]{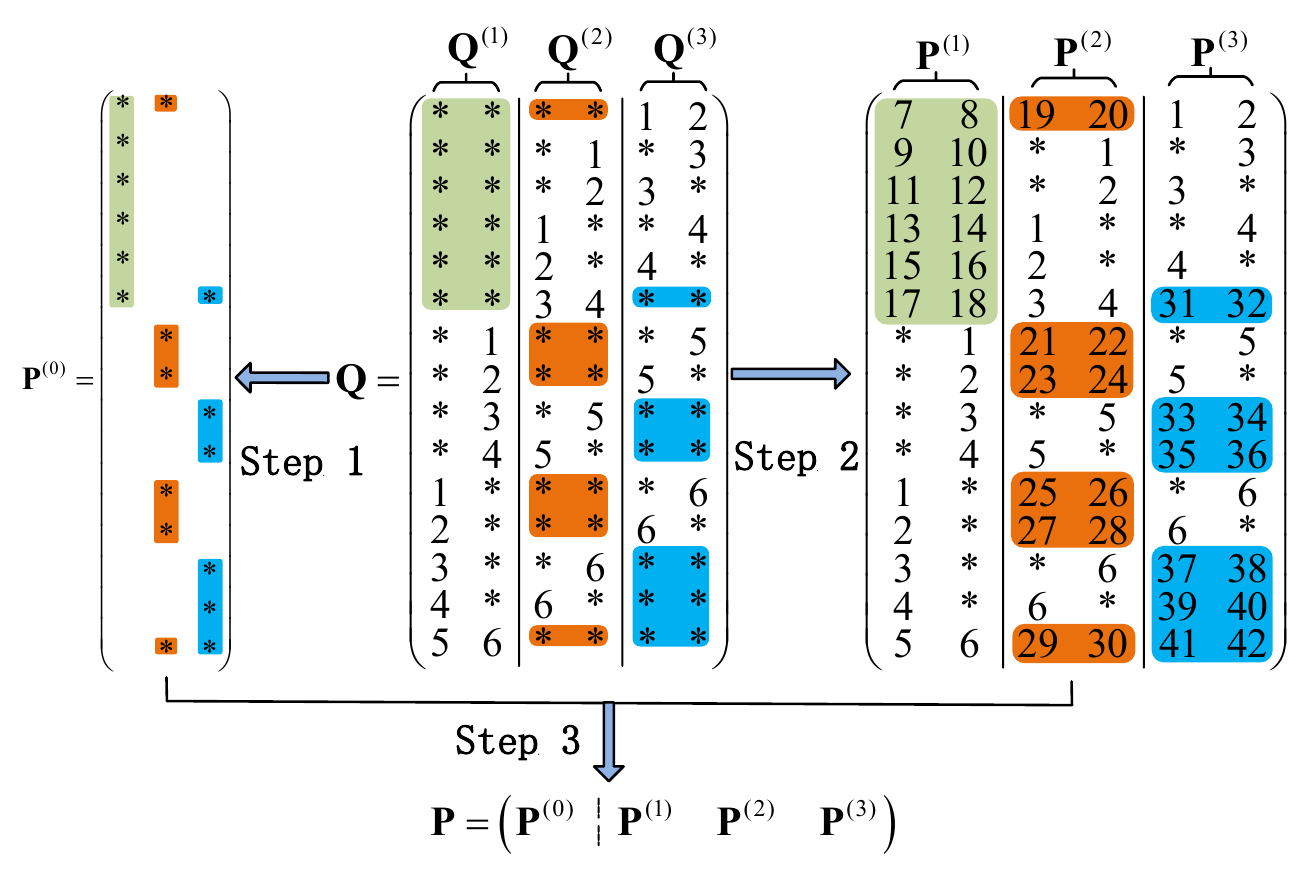}
\caption{ The transformation from MN PDA $\mathbf{Q}$ to a HPDA $\mathbf{P}$ in Theorem \ref{th-maint-2}.}
\label{fig-sketch2}
\end{figure}

\label{ex-sketch}
Given a $(K_1K_2,$$F,Z,S)$$=$$(3\times2,$$15,10,6)$ MN PDA $\mathbf{Q}=\left(\mathbf{Q}^{(1)},\ldots,\mathbf{Q}^{(K_1)}\right)$, we will use a grouping method to construct a $(3,2;15;6,4;$ $\mathcal{S}_\text{m},$$ \mathcal{S}_1,$$ \mathcal{S}_2,$$ \mathcal{S}_3)$ HPDA $\mathbf{P}$ in \eqref{eg-HPDA}, where
\begin{eqnarray}
\label{eq-alphabet-set2}
&&\mathcal{S}_\text{m}=[7:42],\
\mathcal{S}_\text{1}=[1:18],\
\mathcal{S}_\text{2}=[1:6]\cup [19:30],~\\
&&\mathcal{S}_\text{3}=[1:6]\cup [31:42]\nonumber
\end{eqnarray}
 The construction includes the following three steps, as illustrated in Fig. \ref{fig-sketch2}.
\begin{itemize}
\item\textbf{Step 1.} Construction of  $\mathbf{P}^{(0)}$$=(p^{(0)}_{f,k_1})_{f\in[15], k_1\in[3]}$. Because the star rows of $\mathbf{Q}^{(1)}$ are row $1$, $2$, $3$, $4$, $5$, and  $6$, we fill $p^{(0)}_{1,1}=$ $p^{(0)}_{2,1}=$ $p^{(0)}_{3,1}=$ $p^{(0)}_{4,1}=$ $p^{(0)}_{5,1}=$ $p^{(0)}_{6,1}=*$, and for the rest entries in column $1$ of $\mathbf{P}^{(0)}$, we fill them with null. Similarly, the columns $2$ and $3$ of $\mathbf{P}^{(0)}$ can be obtained by applying the same operation on $\mathbf{Q}^{(2)}$ and $\mathbf{Q}^{(3)}$ respectively. Then the resulting array is our required  $\mathbf{P}^{(0)}$. 
\item\textbf{Step 2.} Construction of  $\left(\mathbf{P}^{(1)},\mathbf{P}^{(2)},\mathbf{P}^{(3)}\right)=$ $(p^{(k_1)}_{f,k_2})$ where $f\in[15]$, $k_2\in[2]$, $k_1\in[3]$. Taking $\mathbf{P}^{(1)}$ as an example, we fill the entries in star rows of $\mathbf{Q}^{(1)}$ with distinct integers to get $\mathbf{P}^{(1)}$, i.e., $q^{(1)}_{1,1}=7$, $q^{(1)}_{1,2}=8$, $q^{(1)}_{2,1}=9$, $q^{(1)}_{2,2}=10$, $q^{(1)}_{3,1}=11$, $q^{(1)}_{3,2}=12$, $q^{(1)}_{4,1}=13$, $q^{(1)}_{4,2}=14$, $q^{(1)}_{5,1}=15$, $q^{(1)}_{5,2}=16$, $q^{(1)}_{6,1}=17$, $q^{(1)}_{6,2}=18$. Similarly we can obtain $\mathbf{P}^{(2)}$ and $\mathbf{P}^{(3)}$ in Fig. \ref{fig-sketch2}. The integer sets of $\mathbf{P}^{(1)}$, $\mathbf{P}^{(2)}$ and $\mathbf{P}^{(3)}$ in \eqref{eq-alphabet-set2} can be obtained directly from Fig. \ref{fig-sketch2}.

\item\textbf{Step 3.} Construction of $\mathbf{P}$. We get a $15\times 9$ array by arranging $\mathbf{P}^{(0)}$ and $\mathbf{P}^{(1)}$, $\mathbf{P}^{(2)}$ and $\mathbf{P}^{(3)}$ horizontally, i.e. $\mathbf{P}=\left(\mathbf{P}^{(0)}, \mathbf{P}^{(1)}, \mathbf{P}^{(2)},\mathbf{P}^{(3)}\right)$.
\end{itemize}

Now we   verify that the construction above leads to an HPDA defined in Definition \ref{def-H-PDA}.
\begin{itemize}
\item Each column of $\mathbf{P}^{(0)}$ has $Z_1=6$ stars,  satisfying Condition B$1$.
\item $\mathbf{P}^{(1)}$, $\mathbf{P}^{(2)}$ and $\mathbf{P}^{(3)}$ are $(2,15,4,18)$ PDAs, satisfying    Condition B$2$ of Definition \ref{def-H-PDA}.
\item From Fig. \ref{fig-sketch2}, we have $\mathcal{S}_\text{m}=[7:42]$, whose integers only appear in one $\mathbf{P}^{(k_1)},k_1\in[3]$, and we can check that, if $p^{(k_1)}_{f,k_2}=s\in\mathcal{S}_\text{m}$, then $p^{(0)}_{f,k_1}=*$, thus   Condition B$3$ of Definition \ref{def-H-PDA} holds.
\item It can be checked that Condition B$4$ also  holds. Take $p^{(k_1)}_{f,k_2}=p^{(k'_1)}_{f',k'_2}=1$ as an example.  From \eqref{eg-HPDA}, we   can see that   $p^{(1)}_{11,1}=p^{(1)}_{7,2}=p^{(2)}_{4,1}=p^{(2)}_{2,2}=p^{(3)}_{1,1}=1$. When choosing $f=11,k_1=k_2=1$ and $f'\in\{4,2,1\}$, the corresponding $p^{(0)}_{f',k_1}$ equals to $*$, i.e., $p^{(0)}_{4,1}=p^{(0)}_{2,1}=p^{(0)}_{1,1}=*$, satisfying Condition B$4$ of Definition \ref{def-H-PDA}.
\end{itemize}


\subsection{ General Proof of Theorem \ref{th-maint-2}}
\label{sub-proof-th2}
Given  a $(K,F,Z,S)=$$(K_1K_2,{K_1K_2\choose t},{K_1K_2-1\choose t-1},{K_1K_2\choose t+1})$ MN PDA $\mathbf{Q}=\left(\mathbf{Q}^{(1)}\right.$ $\left.,\ldots,\mathbf{Q}^{(K_1)}\right)$, $t\in[K_2,K_1K_2]$, we show how to construct a $(K_1,K_2$; $F$; $Z_1$, $Z_2$; $\mathcal{S}_\text{m},\mathcal{S}_1$, $\ldots$, $\mathcal{S}_{K_1})=$ $(K_1,K_2$; ${K_1K_2\choose t}$; ${K_1K_2-K_2\choose t-K_2}$, ${K_1K_2-1\choose t-1}-{K_1K_2-K_2\choose t-K_2}$; $\mathcal{S}_\text{m},\mathcal{S}_1$, $\ldots$, $\mathcal{S}_{K_1})$ HPDA  $\mathbf{P}=\left(\mathbf{P}^{(0)}\right.$, $\mathbf{P}^{(1)}$, $\ldots$, $\left.\mathbf{P}^{(K_1)}\right)$, where $\mathcal{S}_{\text{m}}$, $\mathcal{S}_{k_1}$, $k_1\in[K_1]$, are listed in \eqref{eq-alphabet-m} and \eqref{eq-S_k_1} respectively. For the sake of convenience, we use a set $\mathcal{T}\in {[K_1K_2]\choose t}$ to represent the row index of an MN PDA defined in Construction \ref{con-MN} and the constructed HPDA, i.e., $\mathbf{Q}^{(k_1)}$ is represented by $\mathbf{Q}^{(k_1)}=(q^{(k_1)}_{\mathcal{T},k_2})$ $_{\mathcal{T}\in {[K_1K_2]\choose t}, k_2\in[(k_1-1)K_2+1:k_1K_2]}$, $k_1\in[K_1]$.
 The  constructions of $\mathbf{P}^{(0)}$ and  $(\mathbf{P}^{(1)}, \ldots, \mathbf{P}^{(K_1)} )$ are described  as follows:


\begin{itemize}
\item{\bf Step 1.} Construction of $\mathbf{P}^{(0)}$. We construct the $F\times K_1$ mirror site's placement array $\mathbf{P}^{(0)}=(p^{(0)}_{\mathcal{T},k_1})_{\mathcal{T}\in {[K_1K_2]\choose t},k_1\in[K_1]}$ by the following rule:
\begin{eqnarray}
\label{eq-mirror-cache1}
p^{(0)}_{\mathcal{T},k_1}=\left\{\begin{array}{cl}
                  *, & ~\text{if~} q^{(k_1)}_{\mathcal{T},k_2}=*, \begin{array}{c} \forall k_2\in[(k_1-1)K_2+1:k_1K_2]
                                  \end{array}
                                  \\
                  \text{null}, &  \text{otherwise}.
                \end{array}\right.
\end{eqnarray}That is,  let $p^{(0)}_{\mathcal{T},k_1}$ be a star if the     row $\mathcal{T}$ of $\mathbf{Q}^{(k_1)}$  is a star row,  and be null otherwise.

\item {\bf Step 2.} Construction of $(\mathbf{P}^{(1)},\ldots,\mathbf{P}^{(K_1)})$. For each $k_1\in[K_1]$,   $\mathbf{Q}^{(k_1)}$ is used to construct $\mathbf{P}^{(k_1)}=$ $(p^{(k_1)}_{\mathcal{T},k_2})$, where $ \mathcal{T}\in {[K_1K_2]\choose t}, k_2\in[(k_1-1)K_2+1:k_1K_2]$. Note that there are in total $K_1Z_1=K_1{K_1K_2-K_2\choose t-K_2}$ star rows in $\mathbf{Q}^{(1)}$, $\ldots$, $\mathbf{Q}^{(K_1)}$ and each star row has $K_2$ star entries. Then we replace all these star entries in each star row of $\mathbf{Q}^{(1)}$, $\ldots$, $\mathbf{Q}^{(K_1)}$ with consecutive integers from $S+1$ to $S+K_2K_1{K_1K_2-K_2\choose t-K_2}$ to construct $(\mathbf{P}^{(1)},\ldots,\mathbf{P}^{(K_1)})$, and all these integers form the set $\mathcal{S}_{\text{m}}$ as follows.
\begin{eqnarray}
\label{eq-alphabet-m}
\mathcal{S}_{\text{m}}=\left[S+1:\ S+K_2K_1{K-K_2\choose t-K_2}\right]
\end{eqnarray}
\item{\bf Step 3.} Construction of $\mathbf{P}$. We get an $F_1F_2\times (K_1+K_1K_2)$ array by arranging $\mathbf{P}^{(0)}$ and $(\mathbf{P}^{(1)}\ldots, \mathbf{P}^{(K_1)} )$ horizontally, i.e., $\mathbf{P}=\left(\mathbf{P}^{(0)}, \mathbf{P}^{(1)}, \ldots, \mathbf{P}^{(K_1)}\right)$.
\end{itemize}
\subsubsection{Parameter computations} The integer set $\mathcal{S}_{\text{m}}$ can be directly obtained from \eqref{eq-alphabet-m}, which has no intersection with $[S]$.

Then we focus on $\mathcal{S}_{k_1}$. Clearly each $\mathbf{Q}^{(k_1)}$ satisfies Conditions C1 and C3 of Definition \ref{def-PDA}. Now we consider the integer set of $\mathbf{Q}^{(k_1)}$. Recall that $\phi_t(\cdot)$ is a bijection from  $\binom{[K_1K_2]}{t}$ to $[\binom{K_1K_2}{t}]$ in \eqref{Eqn_Def_AN}. Then for any sub-array $\mathbf{Q}^{(k_1)}$, integer $s\in [S]$ is in $\mathbf{Q}^{(k_1)}$ if and only if its inverse mapping $\mathcal{S}=\phi^{-1}_{t+1}(s)$ contains at least one integer of $[(k_1-1)K_2+1:k_1K_2]$, i.e., $\mathcal{S}\cap[(k_1-1)K_2+1:k_1K_2]\neq \emptyset$. So the integer set of $\mathbf{Q}^{(k_1)}$ is
\begin{eqnarray}
\label{eq-S'_k_1}
\begin{split}
&\mathcal{S}'_{k_1}=\bigg\{\phi_{t+1}(\mathcal{S})\ |\mathcal{S}\cap[(k_1-1)K_2+1:k_1K_2]\neq\emptyset,
\mathcal{S}\in {[K_1K_2]\choose t+1}\bigg\}.
\end{split}
\end{eqnarray}
Furthermore, after adding up the integers used for the substitution in Step 2 of $\mathbf{Q}^{(k_1)}$, the integer set of $\mathbf{P}^{(k_1)}$ is
\begin{eqnarray}
\label{eq-S_k_1}
\begin{split}
\mathcal{S}_{k_1}&=\left[S+(k_1-1)K_2{K_1K_2-K_2\choose t-K_2}+1:S+
k_1K_2{K_1K_2-K_2\choose t-K_2}\right] \bigcup \mathcal{S}'_{k_1},&
\end{split}
\end{eqnarray}and $|\mathcal{S}_{k_1}| =$ $ K_2{K_1K_2-K_2\choose t-K_2}+S-{K_1K_2-K_2\choose t+1}$.

\subsubsection{The properties of HPDA verification}
From \eqref{Eqn_Def_AN} the row $\mathcal{T}$ of $\mathbf{Q}^{(k_1)}$ is a star row if and only if all the integers of $[(k_1-1)K_2+1:k_1K_2]$ (i.e.,   indices of users in $\mathcal{U}_{k_1}$) are contained by $\mathcal{T}$. So there are ${K_1K_2-K_2\choose t-K_2}$ star rows, i.e., each column of $\mathbf{P}^{(0)}$ has $Z_1={K_1K_2-K_2\choose t-K_2}$ stars, satisfying  Condition B$1$ of Definition \ref{def-H-PDA}.

From \eqref{eq-S'_k_1} we have
$|\mathcal{S}'_{k_1}|=S-{K_1K_2-K_2\choose t+1}$. So $\mathbf{Q}^{(k_1)}$ is a $(K_2,F,Z,S-{K_1K_2-K_2\choose t+1})$ PDA. Actually, $\mathbf{P}^{(k_1)}$ is obtained by replacing the star entries in star row of $(K_2,F,Z,S-{K-K_2\choose t+1})$ PDA $\mathbf{Q}^{(k_1)}$ by some unique integers which have no intersection with $[S]$. So $\mathbf{P}^{(k_1)}$ is a $(K_2,F,Z_2,S-{K_1K_2-K_2\choose t+1}+K_2{K_1K_2-K_2\choose t-K_2})$ PDA, satisfying Condition B$2$ of Definition \ref{def-H-PDA}.

From Step $2$, we know that each integer $s$ in $\mathcal{S}_{\text{m}}$ occurs only once in $(\mathbf{P}^{(1)},\ldots,\mathbf{P}^{(K_1)})$. Clearly the first part of Condition B$3$ holds. For any $p^{(k_1)}_{\mathcal{T},k_2}=s\in\mathcal{S}_{\text{m}}$, we know row $\mathcal{T}$ is a star row of $\mathbf{Q}^{(k_1)}$, then from \eqref{eq-mirror-cache1} we have $p^{(0)}_{\mathcal{T},k_1}=*$. The second part of  Condition B$3$ holds, then Condition B$3$ holds.

We can show that Condition B$4$ holds by the following reason. Assume that there are two entries $p^{(k_1)}_{\mathcal{T},k_2}=p^{(k^{'}_1)}_{\mathcal{T}^{'},k^{'}_2}=s$, where $k_1\neq k^{'}_1$. Then $s\not\in \mathcal{S}_{\text{m}}$ since each integer of $\mathcal{S}_{\text{m}}$ occurs exactly once. Furthermore if $p^{(k_1)}_{\mathcal{T'},k_2}=s'$ is an integer, then $s'$ must be the element of $\mathcal{S}_{\text{m}}$, otherwise it  contradicts our hypothesis that $\mathbf{Q}$ is a PDA. From the construction of $\mathbf{P}^{(k_1)}$, $s'\in \mathcal{S}_{\text{m}}$ only if the row indexed by $\mathcal{T}$ of $\mathbf{Q}^{(k_1)}$ is a star row. Then from \eqref{eq-mirror-cache1} we have $p^{(0)}_{\mathcal{T},k_1}=*$.

From the above introduction, our expected HPDA is obtained. And the integer set $\bigcup_{k_1=1}^{K_1}\mathcal{S}_{k_1}\setminus\mathcal{S}_{\text{M}}$ is actually the integer set $[S]$, whose cardinality is ${K_1K_2\choose t+1}$. Then by Theorem \ref{th-main-result}, the transmission loads
\begin{eqnarray*}
R_1&=&\frac{|\bigcup_{k_1=1}^{K_1}\mathcal{S}_{k_1}|-|\mathcal{S}_\text{m}|}{F}
=\frac{S}{F}=\frac{K_1K_2-t}{t+1}\\
R_2&=&\max\left.\left\{\  \frac{\mid\mathcal{S}_{k_1}\mid}{F}\ \right| \ k_1\in[K_1] \right\}\\
&=&\frac{K_2{K_1K_2-K_2\choose t-K_2}+S-{K_1K_2-K_2\choose t+1}}{F}\\
&=&\frac{K_1K_2-t}{t+1}-\frac{{K_1K_2-K_2\choose t+1}}{{K_1K_2\choose t}}+\frac{{K_1K_2-K_2\choose t-K_2}K_2}{{K_1K_2\choose t}}
\end{eqnarray*} can be directly obtained.

\section{Proof of Theorem \ref{th-maint-4}}
\label{sec:6}
 In this section, we describe how to construct the $(K_1,K_2$; $F_1F_2$; $Z_1F_2$, $Z_2F_1$; $\mathcal{S}_\text{m}$, $\mathcal{S}_1$, $\ldots$, $\mathcal{S}_{K_1})$ HPDA $\mathbf{P}= \left(\mathbf{P}^{(0)}\right.$, $\mathbf{P}^{(1)}$, $\ldots$, $\left.\mathbf{P}^{(K_1)}\right)$ in Theorem \ref{th-maint-4} based on any $(K_1,F_1,Z_1,S_1)$ PDA $\mathbf{A}$ and $(K_2,F_2,Z_2,S_2)$ PDA $\mathbf{B}$, where
\begin{IEEEeqnarray}{rCl}
&& \mathbf{P}^{(0)}=(p^{(0)}_{(f_1,f_2),k_1}),\ \ \ \ p^{(0)}_{(f_1,f_2),k_1}\in\{*,null\}\nonumber,\\
&& \mathbf{P}^{(k_1)}=(p^{(k_1)}_{(f_1,f_2),k_2}),\ \ \  p^{(k_1)}_{(f_1,f_2),k_2}\in\{*\}\cup \mathcal{S}_{k_1}\nonumber,
\end{IEEEeqnarray} $(f_1,f_2)$ is a couple indicating the $(f_1-1)F_1+f_2$-th row and where $f_1\in[F_1]$, $f_2\in[F_2]$, $k_1\in [K_1]$, $k_2\in[K_2]$.
The construction of HPDA $\mathbf{P}= \left(\mathbf{P}^{(0)}\right.$, $\mathbf{P}^{(1)}$, $\ldots$, $\left.\mathbf{P}^{(K_1)}\right)$ for Theorem \ref{th-maint-4} can be briefly described as follows: Given two PDAs, denoted by $\mathbf{A}$, $\mathbf{B}$ respectively, $\mathbf{P}^{(0)}$ is obtained by deleting all the integers of $\mathbf{A}$ and then simply expanding each row of it, and
$\left(\mathbf{P}^{(1)}\right.$, $\ldots$, $\left.\mathbf{P}^{(K_1)}\right)$ is obtained by using a hybrid method, in which $\mathbf{A}$ acts as an outer array and the inner arrays are simply constructed from $\mathbf{B}$.

In the follows, we first give an illustrative example to show the  construction of $\mathbf{P}$ based on two MN PDAs, and then present our general proof of coded caching scheme based on the HPDA in Theorem \ref{th-maint-4}.
\subsection{Example of the Construction of HPDA in Theorem \ref{th-maint-4}}
\label{sub-sketch}
Given a $(K_1,F_1,Z_1,S_1)=(2,2,1,1)$ MN PDA $\mathbf{A}=(a_{f_1,k_1})_{f_1\in[2], k_1\in[2]}$ and a $(K_2,F_2,Z_2,S_2)=(3,3,1,3)$ MN PDA $\mathbf{B}=(b_{f_2,k_2})_{f_2\in[3], k_2\in[3]}$ where
\begin{eqnarray}
\label{eq-three-array}
\mathbf{A}=\left(
\begin{array}{cc}
* & 1 \\
1 & *
\end{array}
\right) \ \ , \ \
\mathbf{B}=\left(
\begin{array}{ccc}
*	&	1	&	2\\
1	&	*	&	3\\
2	&	3	&	*
\end{array}
\right).
\end{eqnarray}
 We will use a hybrid method to construct a $(2,3;6;3,2;$ $\mathcal{S}_\text{m}$, $\mathcal{S}_\text{1},\mathcal{S}_\text{2})$ HPDA where
\begin{eqnarray}
\label{eq-alphabet-set}
\mathcal{S}_\text{m}=[4:9],\ \
\mathcal{S}_\text{1}=[1:6],\ \
\mathcal{S}_\text{2}=[1:3]\cup [7:9],
\end{eqnarray}
 through the following three steps, as illustrated in Fig. \ref{fig-sketch}.
\begin{figure}[http!]
\centering
\includegraphics[scale=0.9]{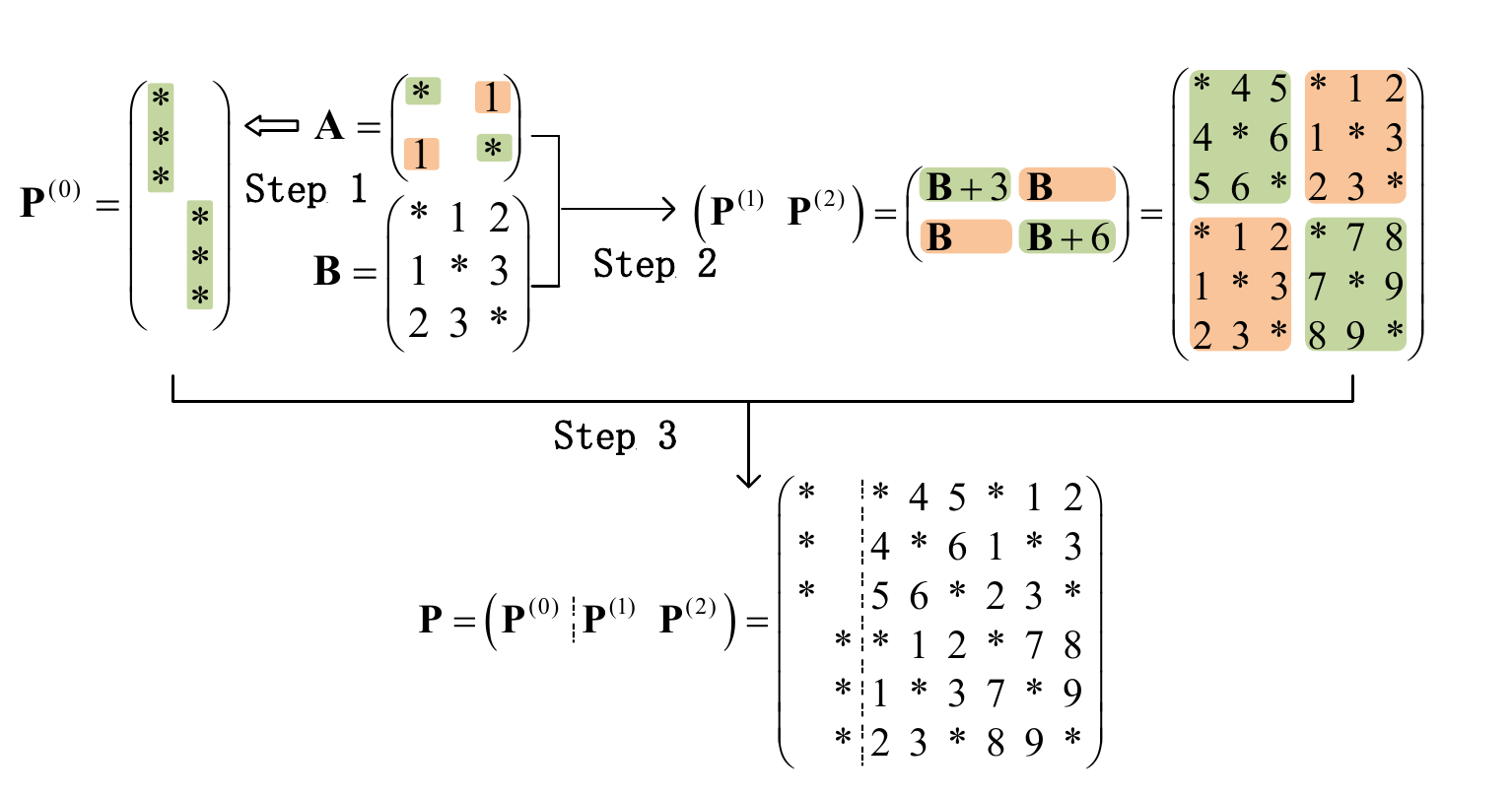}
\caption{ The transformation from MN PDAs $\mathbf{A}$ and $\mathbf{B}$ to a HPDA $\mathbf{P}$ in Theorem \ref{th-maint-4}.}
\label{fig-sketch}
\end{figure}
\begin{itemize}
\item\textbf{Step 1.} Construction of $\mathbf{P}^{(0)}$ for mirror sites. We can get a $6\times 2$ array $\mathbf{P}^{(0)}$ by deleting all the integer entries of $\mathbf{A}=(a_{f_1,k_1})_{f_1\in[2], k_1\in[2]}$ and then expanding each row $3$ times\footnote{\label{foot;ICM(20,2)}
		This is the row number of the inner structure array which will be introduced in Step 2.}.
\item\textbf{Step 2.} Construction of $\left(\mathbf{P}^{(1)},\mathbf{P}^{(2)}\right)$ for users. We replace the integer entries $a_{2,1}$ and $a_{1,2}$ by $\mathbf{B}$, and replace $a_{1,1}=a_{2,2}=*$ by $\mathbf{B}+3$, $\mathbf{B}+6$ respectively to get
    \begin{eqnarray}
    \label{eq-subarrays}
    \begin{split}
\mathbf{P}^{(1)}=\left(
\begin{array}{c}
\mathbf{B}+3\\
\mathbf{B}
\end{array}
\right),\ \
\mathbf{P}^{(2)}=\left(
\begin{array}{c}
\mathbf{B}\\
\mathbf{B}+6
\end{array}
\right).
\end{split}
\end{eqnarray}
\item\textbf{Step 3.} Construction of $\mathbf{P}$. We get a $6\times 8$ array by arranging $\mathbf{P}^{(0)}$ and $(\mathbf{P}^{(1)},\mathbf{P}^{(2)})$ horizontally, i.e., $\mathbf{P}=\left(\mathbf{P}^{(0)},\mathbf{P}^{(1)},\mathbf{P}^{(2)}\right)$.
\end{itemize}
Now we verify that the construction above  leads to a HPDA defined in Definition \ref{def-H-PDA}.
\begin{itemize}
\item Each column of $\mathbf{P}^{(0)}$ has $Z_1=3$ stars,  satisfying Condition B$1$ of Definition \ref{def-H-PDA}.
\item $\mathbf{P}^{(1)}$ and $\mathbf{P}^{(2)}$ are $(3,6,2,6)$ PDAs, satisfying Condition B$2$ of Definition \ref{def-H-PDA}.
\item From Fig. \ref{fig-sketch}, we have $\mathcal{S}_\text{m}=[4,9]$, whose integers only appear in one $\mathbf{P}^{(k_1)},k_1\in[2]$, and we can check that, if $p^{(k_1)}_{(f_1,f_2),k_2}=s\in\mathcal{S}_\text{m}$, then $p^{(0)}_{(f_1,f_2),k_1}=*$, thus   Condition B$3$ of Definition \ref{def-H-PDA} holds.
\item Finally we claim that the Condition B$4$ of Definition \ref{def-H-PDA} holds. Here we take $p^{(k_1)}_{(f_1,f_2),k_2}$$=p^{(k'_1)}_{(f'_1,f'_2),k'_2}$$=1$ as an example. From Fig. \ref{fig-sketch}, we   can see that $p^{(1)}_{(2,2),1}$ $=p^{(1)}_{(2,1),2}$ $=p^{(2)}_{(1,2),1}$ $=p^{(2)}_{(1,1),2}=1$. When choosing $(f_1,f_2)=(2,2)$, $k_1=k_2=1$ and $(f'_1,f'_2)\in\{(1,2),(1,1)\}$, the corresponding $p^{(0)}_{(f'_1,f'_2),k_1}$ equals to $*$, i.e., $p^{(0)}_{(1,2),1}=p^{(0)}_{(1,1),1}=*$, satisfying Condition B$4$ of Definition \ref{def-H-PDA}.


\end{itemize}

\subsection{ General Proof of Theorem \ref{th-maint-4}}
\label{proof:Th3}
In this subsection, we will prove Theorem \ref{th-maint-4} by constructing a $(K_1,K_2$; $F_1F_2$; $Z_1F_2$, $Z_2F_1$;  $\mathcal{S}_\text{m},\mathcal{S}_{1},\ldots,\mathcal{S}_{K_1})$ HPDA $\mathbf{P}= \left(\mathbf{P}^{(0)}\right.$, $\mathbf{P}^{(1)}$, $\ldots$, $\left.\mathbf{P}^{(K_1)}\right)$ with any $(K_1,F_1,Z_1,S_1)$ PDA $\mathbf{A}=(a_{f_1,k_1})_{f_1\in[F_1],\ k_1\in[K_1]}$ and $(K_2,F_2,Z_2,S_2)$ PDA $\mathbf{B}=(b_{f_2,k_2})_{f_2\in[F_2],\ k_2\in[K_2]}$, where $\mathcal{S}_\text{m}$ and $\mathcal{S}_{k_1}$ will be proposed in \eqref{eq-s_m} and \eqref{eq-S-k_1} respectively. The constructions of $\mathbf{P}^{(0)}$ and  $\left(\mathbf{P}^{(1)}\ldots, \mathbf{P}^{(K_1)} \right)$ are described  as follows:
\begin{itemize}
\item{\bf Step 1.} Construction of $\mathbf{P^{(0)}}$.
We can get an $F_1F_2\times K_1$ array $\mathbf{P}^{(0)}$ by deleting all the integers in $\mathbf{A}$ and expanding each row by $F_2$ times.  Then each entry $p^{(0)}_{(f_1,f_2),k_1}$ can be written as follows.
\begin{eqnarray}
\label{eq-mirror-cache}
p^{(0)}_{(f_1,f_2),k_1}=\left\{\begin{array}{cc}
                  * & \text{if } a_{f_1,k_1}=*,\\
                  \text{null} & \text{otherwise}.
                \end{array}\right.
\end{eqnarray} 
\item{\bf Step 2.} Construction of $\left(\mathbf{P}^{(1)}\ldots, \mathbf{P}^{(K_1)} \right)$.
The main idea of constructing $\left(\mathbf{P}^{(1)}\ldots, \mathbf{P}^{(K_1)} \right)$ is replacing the entries of $\mathbf{A}$ by inner array $\mathbf{B}$ and adjusting its integers. As $\mathbf{A}$ consists of integer-type entries and star-type entries, our construction consists of the following two parts. Firstly we replace each integer entry $a_{f_1,k_1}=s$ by an $F_2\times K_2$ array
\begin{eqnarray}
\label{eq-type-I}
\mathbf{I}_1(s)=\mathbf{B}+(s-1)\times S_2.
\end{eqnarray} Secondly we replace each star entry $a_{f_1,k_1}=*$ by an $F_2\times K_2$ array
\begin{eqnarray}
\label{eq-type-II}
\mathbf{I}_2(k_1,f_1)=\mathbf{B}+[(k_1-1)Z_1+\varphi_{k_1}(f_1)-1+S_1]\times S_2,
\end{eqnarray} where $\varphi_{k_1}(f_1)$ represents the order of the row labels from up to down among all the star entries in $k_1$-th column of $\mathbf{A}$.

 Here we take  $\left(\mathbf{P}^{(1)}, \mathbf{P}^{(2)} \right)$ in Fig. \ref{fig-sketch} as an example. As $a_{1,2}=a_{2,1}=s=1$, we have $\mathbf{I}_1(1)=\mathbf{B}+(1-1)\times3=\mathbf{B}$ in \eqref{eq-three-array}. While $a_{1,1}=a_{2,2}=*$, $\varphi_{1}(1)=\varphi_{2}(2)=1$, we have $\mathbf{I}_2(1,1)=\mathbf{B}+[(1-1)\times1+1-1+1]\times3=\mathbf{B}+3$ and $\mathbf{I}_2(2,2)=\mathbf{B}+[(2-1)\times1+1-1+1]\times3=\mathbf{B}+6$ in \eqref{eq-two-array}, where
\begin{eqnarray}
\label{eq-two-array}
\mathbf{B}+3=\left(
\begin{array}{ccc}
*	&	4	&	5\\
4	&	*	&	6\\
5	&	6	&	*
\end{array}
\right),
\mathbf{B}+6=\left(
\begin{array}{ccc}
*	&	7	&	8\\
7	&	*	&	9\\
8	&	9	&	*
\end{array}
\right).
\end{eqnarray}
From the above example, obviously the integer sets in $\mathbf{I}_1(1)=\mathbf{B}$, $\mathbf{I}_2(1,1)=\mathbf{B}+3$ and $\mathbf{I}_2(2,2)=\mathbf{B}+6$ have no common integer, and we can generalize the investigation to the general cases, as illustrated in Lemma \ref{lem-integer-se-order}.

\begin{lemma}
\label{lem-integer-se-order}
For any integers $k_1$, $k_1'\in[K_1]$, $f_1$, $f_1'\in[F_1]$, $s$ and $s'$, we have the following statements on the integer sets in $\mathbf{I}_1(s)$, $\mathbf{I}_1(s')$, $\mathbf{I}_2(k_1,f_1)$ and $\mathbf{I}_2(k_1',f_1')$.
\begin{itemize}
\item 
    The integer sets in $\mathbf{I}_1(s)$ and $\mathbf{I}_1(s')$ have no common integer if and only if $s\neq s'$;
\item 
    When $k_1=k'_1$, the integer sets in $\mathbf{I}_2(k_1,f_1)$ and $\mathbf{I}_2(k'_1,f'_1)$ have no common integer if and only if $\varphi_{k_1}(f_1)\neq\varphi_{k'_1}(f'_1)$;
\item 
    When $k_1\neq k'_1$, the integer sets in $\mathbf{I}_2(k_1,f_1)$ and $\mathbf{I}_2(k'_1,f'_1)$ have no common integer;
\item 
    When $s\in[S_1]$, the integer sets in $\mathbf{I}_1(s)$ and $\mathbf{I}_2(k_1,f_1)$ have no common integer;
\end{itemize}
\end{lemma}
By Lemma \ref{lem-integer-se-order}, each entry $p^{(k_1)}_{(f_1,f_2),k_2}$ in $\left(\mathbf{P}^{(1)}\ldots, \mathbf{P}^{(K_1)} \right)$ is determined uniquely and can be written as follows,
\begin{IEEEeqnarray}{rCl}
\label{eq-user-cache}
p^{(k_1)}_{(f_1,f_2),k_2}=\left\{\begin{array}{l}
                  b_{f_2,k_2}\!+\!(s\!-\!1)S_2 ,\quad \text{if } a_{f_1,k_1}\!=\!s,\\
                  b_{f_2,k_2}\!+\!\big[(k_1\!-\!1)Z_1\!
                 + \varphi_{k_1}(f_1)\!-\!1\!+\!S_1\big]S_2 , \  \text{if } a_{f_1,k_1}=*.
                \end{array}\right.
\end{IEEEeqnarray}

\item{\bf Step 3.} Construction of $\mathbf{P}$.
\label{sec:con-hpda}
We get an $F_1F_2\times$$(K_1\!+\!K_1K_2)$ array by arranging $\mathbf{P}^{(0)}$ and $\left(\mathbf{P}^{(1)},\ldots,\mathbf{P}^{(K_1)}\right)$ horizontally, i.e., $\mathbf{P}=\left(\mathbf{P}^{(0)}\right.$,$\mathbf{P}^{(1)}$,$ \ldots$,$\left. \mathbf{P}^{(K_1)}\right)$.
\end{itemize}

\subsubsection{Parameter computations}
\label{subsub-parameters}
For a $(K,F,Z,S)$ PDA, we define $\mathcal{C}_{i}$ as the integer set containing all the integers in the $i$-th column where $\mathcal{C}_{i}\subset[S]$ and  $|\mathcal{C}_{i}|=F-Z$. Firstly we consider the integer set $\mathcal{S}_{\text{m}}$, which is the union set of the integer set of each $\mathbf{I}_2(k_1,f_1)$. There are in total $K_1Z_1$ stars in $\mathbf{A}$, then from $\eqref{eq-type-II}$ and Lemma \ref{lem-integer-se-order} we have
\begin{eqnarray}
\label{eq-s_m}
\begin{split}
\mathcal{S}_{\text{m}}&=\left(S_1S_2:\ (S_1+Z_1K_1)S_2\right].&
\end{split}
\end{eqnarray}
Obviously the cardinality of $\mathcal{S}_{\text{m}}$ is $|\mathcal{S}_{\text{m}}|=Z_1K_1S_2$. 

Secondly we focus on $\mathcal{S}_{k_1}$, i.e., the integer set of $\mathbf{P}^{(k_1)}$ for each $k_1\in[K_1]$. From \eqref{eq-type-I}, $\mathbf{P}^{(k_1)}$ is composed of $F_1-Z_1$ inner arrays $\mathbf{I}_1(s)$ and $Z_1$ inner arrays $\mathbf{I}_2(k_1,f_1)$. So the integer set $\mathcal{S}_{k_1}$ is actually the union set of all integer sets of the $F_1$ inner arrays. Then from \eqref{eq-type-I}, \eqref{eq-type-II} and Lemma \ref{lem-integer-se-order},  the integer set of $\mathbf{P}^{(k_1)}$ is
\begin{eqnarray}
\label{eq-S-k_1}
\begin{split}
\mathcal{S}_{k_1}&=\left(((k_1-1)Z_1+S_1)S_2:(k_1Z_1+S_1)S_2\right]
\bigcup\left(\bigcup\limits_{s\in \mathcal{C}_{k_1},k_1\in [K_1] }\left(0+(s-1)S_2:sS_2\right]\right),
\end{split}
\end{eqnarray} $k_1\in [K_1]$.  Because $s\in[S_1]$, according to Lemma \ref{lem-integer-se-order}, all the $F_1$ integer sets of inner arrays used for composing $\mathbf{P}^{(k_1)}$ do not have any common integer, so we have $|\mathcal{S}_{k_1}|=F_1S_2$.
\subsubsection{The properties of HPDA verification}
\label{subsub-verify}
Because there are $Z_1$ stars in each column of $\mathbf{A}$, from \eqref{eq-mirror-cache} each column of $\mathbf{P}^{(0)}$ has exactly $Z_1F_2$ stars, satisfying   Condition B$1$ of Definition \ref{def-H-PDA}.

Then we focus on Condition B$2$, i.e.,  $\mathbf{P}^{(k_1)}$ is a $(K_2$, $F_1F_2$, $F_1Z_2$, $ F_1S_2)$ PDA. Because $\mathbf{P}^{(k_1)}$ is composed of $F_1$ inner arrays, each column of $\mathbf{P}^{(k_1)}$ has $F_1Z_2$ stars. So Condition C$1$ of Definition \ref{def-PDA} holds. In the above we have $|\mathcal{S}_{k_1}|=F_1S_2$, obviously C$2$ of Definition \ref{def-PDA} holds.  Because all the arrays defined in \eqref{eq-type-I} and \eqref{eq-type-II} satisfy Condition C$3$ of Definition \ref{def-PDA} and by Lemma \ref{lem-integer-se-order}, the intersection of the integer sets of any two of the $F_1$ inner arrays is empty, then each $\mathbf{P}^{(k_1)}$ also satisfies the Condition C$3$. Thus,  each $\mathbf{P}^{(k_1)}$ is a $(K_2,F_1F_2,F_1Z_2,F_1S_2)$ PDA. 

Now  consider Condition B$3$. Recall that all the integers in $\mathcal{S}_{\text{m}}$ are generated from \eqref{eq-type-II}. If $k_1\neq k'_1$, by the third statement of Lemma \ref{lem-integer-se-order}, each integer in $\mathcal{S}_{\text{m}}$ only exists in one $\mathbf{P}^{(k_1)}$. When the entry  $p^{(k_1)}_{(f_1,f_2),k_2}=s\in \mathcal{S}_{\text{m}}$, from \eqref{eq-user-cache} and \eqref{eq-mirror-cache} we have $a_{f_1,k_1}=*$, $p^{(0)}_{(f_1,f_2),k_1}=*$. Thus, Condition B$3$ holds.

Finally we consider the Condition B$4$. For any integers $k_1$, $k_1'\in [K_1]$, $k_2$, $k_2'\in [K_2]$ and any couples $(f_1,f_2)$, $(f_1',f_2')$, assume that $p^{(k_1)}_{(f_1,f_2),k_2}=
p^{(k'_1)}_{(f_1',f_2'),k'_2}=s$ is an integer. By Lemma \ref{lem-integer-se-order}, the case $k_1\neq k'_1$, $f_1=f'_1$ is impossible since the intersection of the integer sets of related inner arrays is empty. So we only need to consider the case where $k_1\neq k'_1$, $f_1\neq f'_1$, and there are three conditions:
\begin{itemize}
\item $p^{(k_1)}_{(f_1,f_2),k_2}$ and $p^{(k_1')}_{(f_1',f_2'),k_2'}$ are all in the arrays generated by \eqref{eq-type-II}. By the third statement of Lemma \ref{lem-integer-se-order}, this case is impossible since the intersection of the integer sets in related inner arrays is empty.
\item $p^{(k_1)}_{(f_1,f_2),k_2}$ and $p^{(k_1')}_{(f_1',f_2'),k_2'}$ are in the arrays generated by \eqref{eq-type-I} and \eqref{eq-type-II} respectively. By the forth statement of Lemma \ref{lem-integer-se-order}, this case is also impossible since the intersection of the integer sets in related inner arrays is empty.
\item $p^{(k_1)}_{(f_1,f_2),k_2}$ and $p^{(k_1')}_{(f_1',f_2'),k_2'}$ are all in the arrays generated by \eqref{eq-type-I}. Then we have $s=b_{f_2,k_2}+(s'-1)\times S_2=b_{f_2',k_2'}+(s''-1)\times S_2$ and it's     true if and only if $b_{f_2,k_2}=b_{f_2',k_2'}$, $a_{f_1,k_1}=s'=a_{f'_1,k'_1}=s''$, because $b_{f_2,k_2}, b_{f_2',k_2'}\leq S_2$. Without loss of generality we assume that $p^{(k_1)}_{(f'_1,f'_2),k_2}$ is an integer entry. Because $a_{f_1,k_1}=a_{f'_1,k'_1}=s'$, and from Condition C$3$ of definition \ref{def-PDA} we have $a_{f_1',k_1}=a_{f_1,k_1'}=*$. According to \eqref{eq-mirror-cache} we have $p^{(0)}_{(f_1',f_2'),k_1}=*$.
\end{itemize}
 From the above discussion, the Condition B$4$ of Definition \ref{def-H-PDA} holds. Thus, $\mathbf{P}$ is our expected HPDA.

 The integer set $\bigcup_{k_1=1}^{K_1}\mathcal{S}_{k_1}\setminus\mathcal{S}_{\text{M}}$ is the union set of  integer sets in $S_1$ inner arrays $\mathbf{I}_1(s)$, $s\in[S_1]$, then we have $|\bigcup_{k_1=1}^{K_1}\mathcal{S}_{k_1}\setminus\mathcal{S}_{\text{M}}|=S_1S_2$. From Theorem \ref{th-main-result} we have the load for the first layer
\begin{eqnarray*}
&R_1=\frac{|\bigcup_{k_1=1}^{K_1}\mathcal{S}_{k_1}|-|\mathcal{S}_\text{m}|}{F}
=\frac{S_1S_2}{F_1F_2}
\end{eqnarray*} and the load for the second layer
\begin{eqnarray*}
R_2&=&\max\left.\left\{\  \frac{\mid\mathcal{S}_{k_1}\mid}{F}\ \right| k_1\in[K_1]  \right\}\\
&=&\frac{F_1S_2}{F_1F_2}\\
&=&\frac{S_2}{F_2}.
\end{eqnarray*}
\section{Conclusion}
\label{sec:7}
In this paper, we studied the hierarchical network model and introduced a new combination structure, referred as HPDA, which can be used to characterize both the placement and delivery strategy of the coded caching scheme. So the problem of designing a scheme for hierarchical network is transformed into constructing an appropriate HPDA. Firstly we propose a class of HPDAs, which achieves the lower bound of the first layer transmission load $R_1$ for non-trivial cases, by dividing the MN PDAs into several equal size groups. Due to the limitation of the system parameters in this class of HPDAs,  we then proposed another class of HPDAs via a hybrid construction of two PDAs. Consequently, using any two PDAs, a new HPDA can be obtained which   allows   flexible system parameters and has a smaller subpacketization level compared with our first class of HPDAs.

\begin{appendices}
\section{Lower Bound of $R^*_1$}
\label{appendix-optimal}
Recall that during the data placement placement, the cached contents at user $\mathbf{U}_{k_1,k_2}$ and mirror site $k_1$ are  $\mathcal{Z}_{k_1}$ and $\mathcal{Z}_{(k_1,k_2)}$, respectively.

Now we introduce an enhanced system where   each user already knows the cache contents of its connected mirror site. For this enhanced system, denote the   cache contents of the user $\mathbf{U}_{k_1,k_2}$  as  $\bar{\mathcal{Z}}_{(k_1,k_2)}= \mathcal{Z}_{k_1}\cup \mathcal{Z}_{(k_1,k_2)} $. In   uncoded placement scenarios,      each file can be viewed as a collection of   $2^{K_1K_2}$ packets as $W_{i}=\{W_{i,\mathcal{T}}|\mathcal{T}\subseteq[K_1]\times[K_2]\}$, where user  $\mathbf{U}_{k_1,k_2}$ stores $W_{i,\mathcal{T}} $ if $(k_1,k_2)\in\mathcal{T}$. Consider one permutation of $[K_1]\times[K_2]$ denoted by $\{(1,1),(1,2),\ldots,(K_1,K_2)\}$, $(k_1,k_2)\in[K_1]\times[K_2]$, and one demand vector $\mathbf{d}=\{d_{1,1}, d_{1,2},\ldots,d_{K_1,K_2}\}$ where $d_{k_1,k_2}\neq d_{k'_1,k'_2}$ if $k_1\neq k'_1$ or $k_2\neq k'_2$.     We then   construct a genie-aided super-user with cached content
\begin{IEEEeqnarray}{rCl}
\bar{ {Z}}~&&=(\bar{\mathcal{Z}}_{(1,1)},\bar{\mathcal{Z}}_{(1,2)}\backslash(\bar{\mathcal{Z}}_{(1,1)}\cup W_{d_{1,1}}),\ldots,\nonumber\\
&&\quad\quad \bar{\mathcal{Z}}_{(K_1,K_2)}\backslash(\bar{\mathcal{Z}}_{(1,1)}\cup W_{d_{1,1}}\cup
\bar{\mathcal{Z}}_{(1,2)}\cup W_{d_{1,2}}\cup\cdots
\cup \bar{\mathcal{Z}}_{(K_1,K_2-1)}\cup W_{d_{K_1,K_2-1}}))).
\end{IEEEeqnarray}

The genie-aided super-user is able to recover $W_{d_{1,1}},$ $ W_{d_{1,2}},\ldots,$ $W_{d_{K_1,K_2}}$  from $(X,X_1,\ldots,X_{K_1},\bar{Z})$, where $X$ and $X_{k_1}$ are the signals sent by the server and mirror site $k_1$, respectively. Thus, we have
 \begin{IEEEeqnarray}{rCl}
&& H(W_{d_{1,1}}, W_{d_{1,2}},\ldots,W_{d_{K_1,K_2}}|\bar{Z}) \nonumber\\
&&\quad  =  H(W_{d_{1,1}}, W_{d_{1,2}},\ldots,W_{d_{K_1,K_2}}|X,X_1,\ldots,X_{K_1},\bar{Z})
+I(W_{d_{1,1}}, W_{d_{1,2}},\ldots,W_{d_{K_1,K_2}};X,X_1,\ldots,X_{K_1}|\bar{Z})\nonumber\\
&&\quad  =I(W_{d_{1,1}}, W_{d_{1,2}},\ldots,W_{d_{K_1,K_2}};X,X_1,\ldots,X_{K_1}|\bar{Z})\nonumber\\
&&\quad  \leq H(X,X_1,\ldots,X_{K_1}|\bar{Z})\nonumber\\
&&\quad   {=} H(X|\bar{Z})  \label{barZup}
\end{IEEEeqnarray}
 where  the last equality holds because $\bar{Z}$  contains all mirrors' contents $\mathcal{Z}_{1},\ldots,\mathcal{Z}_{K_1}$, leading to $H(X_{k_1}|X,\bar{Z})=H(X_{k_1}|X,\mathcal{Z}_{k_1})=0$ for all $k_1\in[K_1]$.

 Next we introduce a more powerful enhanced system    where   each user $\mathbf{U}_{k_1,k_2}$ has a caching size of $(M_1+M_2)B$ bits, and denote its cached content as   $\hat{\mathcal{Z}}_{(k_1,k_2)}$. Note that this enhanced system can only result in smaller communication loads in $R_1$ and $R_2$  than that of the first enhanced   system. This is because  in the new enhanced system  each users $\mathbf{U}_{k_1,k_2} $ is able to   cache any set of  sub-files of   $(M_1+M_2)B$ bits, including the caching strategy of the first enhanced system $\bar{\mathcal{Z}}_{(k_1,k_2)}= \mathcal{Z}_{k_1}\cup \mathcal{Z}_{(k_1,k_2)} $. We then construct a new genie-aided super-user with cached content
\begin{IEEEeqnarray}{rCl}\label{contentbhatZ}
\hat{ {Z}}~&&=(\hat{\mathcal{Z}}_{(1,1)},\hat{\mathcal{Z}}_{(1,2)}\backslash(\hat{\mathcal{Z}}_{(1,1)}\cup W_{d_{1,1}}),\ldots,\nonumber\\
&&\ \ \ \bar{\mathcal{Z}}_{(K_1,K_2)}\backslash(\hat{\mathcal{Z}}_{(1,1)}\cup W_{d_{1,1}}\cup \hat{\mathcal{Z}}_{(1,2)}\cup W_{d_{1,2}}\cup\cdots
\cup\hat{\mathcal{Z}}_{(K_1,K_2-1)}\cup W_{d_{K_1,K_2-1}}))).
\end{IEEEeqnarray}
Due to the stronger caching ability of the new genie-aided super-user, we have
 \begin{IEEEeqnarray}{rCl}\label{keystepLower}
 H(W_{d_{1,1}}, W_{d_{1,2}},\ldots,W_{d_{K_1,K_2}}|\hat{Z})
 \leq  H(W_{d_{1,1}}, W_{d_{1,2}},\ldots,W_{d_{K_1,K_2}}|\bar{Z})
 \stackrel{(a)}{\leq} H(X|\bar{Z})\leq H(X).\quad
\end{IEEEeqnarray}
 where   (a) holds by \eqref{barZup}. From \eqref{keystepLower}, we obtain that
 \begin{IEEEeqnarray}{rCl}
 R_1\geq\sum_{(k_1,k_2)\in[K_1]\times[K_2]} \ \ \sum_{\mathcal{T}\subseteq[K_1]\times[K_2] \backslash\{(1,1),\ldots,(k_1,k_2)\}} \frac{W_{d_{k_1,k_2},\mathcal{T}}}{B}\nonumber\\
 &&
 \label{Rlower1}
\end{IEEEeqnarray}
 This is equivalent to a single-layer coded caching system where the server connects  $K$-user each equipped with cache memory of $(M_1+M_2)B$ bits. Now we follow the method in \cite{Wang'ITW16} to prove the lower bound of $R^*_1$.

Summing all the inequalities in the form of \eqref{Rlower1} over all permutations of
users and all demand vectors in which users have distinct demands, we obtain that
\begin{subequations}\label{Rlower2}
\begin{IEEEeqnarray}{rCl}
R_1\geq \sum_{t\in[0:K_1K_2]} \frac{\binom{K_1K_2}{t+1}}{N\binom{K_1K_2}{t}}x_t = \sum_{t\in[0:K_1K_2]} \frac{K_1K_2-t}{t+1} x_t
\end{IEEEeqnarray}
where
\begin{IEEEeqnarray}{rCl}
x_t=\sum_{i\in[N]}\ \ \sum_{\mathcal{T}\subseteq[K_1]\times[K_2]:|\mathcal{T}|=t} \frac{W_{i,\mathcal{T}}}{B}.
\end{IEEEeqnarray}
\end{subequations}
Also, we have the following conditions   due to the constraints on the file size and memory size
\begin{IEEEeqnarray}{rCl}\label{Rlower2Condition}
\sum_{t\in[0:K_1K_2]} x_t =N,\quad  \sum_{t\in[0:K_1K_2]}tx_t=K_1K_2(M_1+M_2).
\end{IEEEeqnarray}
Combining \eqref{Rlower2} and \eqref{Rlower2Condition} and by Fourier Motzkin elimination, we obtain  the lower bound of $$R_1^*\geq\frac{K_1K_2-t}{t+1}, ~\text{for~}  t\in[0:K_1K_2] . $$

\section{Proof of Lemma \ref{lem-integer-se-order}}\label{appen-set-order}
Without loss of generality we assume that $s'>s$. From \eqref{eq-type-I} the first statement holds since the minimum integer of the integer set of $\mathbf{I}_1(s')$ minus the maximum integer of the integer set of $\mathbf{I}_1(s)$ is
\begin{eqnarray*}
&&\left(1+(s'-1)S_2 \right)-\left(S_2+(s-1)S_2\right)\\
&&\ =[s'-(s+1)]S_2+1\\
&&\ \geq1.
\end{eqnarray*} While if $s'=s$, $\mathbf{I}_1(s)$ and $\mathbf{I}_1(s')$ are the same array.

Without loss of generality we assume that $\varphi_{k'_1}(f'_1)>\varphi_{k_1}(f_1)$. From \eqref{eq-type-II} the second statement holds since the minimum integer of the integer set of $\mathbf{I}_2(k'_1,f'_1)$ minus the maximum integer of the integer set of $\mathbf{I}_2(k_1,f_1)$ is
\begin{eqnarray*}
&&1+[(k'_1-1)Z_1+\varphi_{k'_1}(f'_1)-1+S_1]S_2-S_2-
[(k_1-1)Z_1+\varphi_{k_1}(f_1)-1+S_1]S_2\\
&&\ =(\varphi_{k'_1}(f'_1)-\varphi_{k_1}(f_1)-1)S_2+1\\
&&\ \geq 1.
\end{eqnarray*} While if $\varphi_{k'_1}(f'_1)=\varphi_{k_1}(f_1)$, $\mathbf{I}_2(k'_1,f'_1)$ and $\mathbf{I}_2(k_1,f_1)$ are the same array.

Without loss of generality we assume that $k'_1>k_1$. From \eqref{eq-type-II} the third statement holds since the minimum integer of the integer set of $\mathbf{I}_2(k'_1,f'_1)$ minus the maximum integer of the integer set of $\mathbf{I}_2(k_1,f_1)$ is
\begin{eqnarray*}
&&1+[(k'_1-1)Z_1+\varphi_{k'_1}(f'_1)-1+S_1]S_2-S_2-[(k_1-1)Z_1+\varphi_{k_1}(f_1)-1+S_1]S_2\\
&&\ =1+[(k'_1-k_1)Z_1+\varphi_{k'_1}(f'_1)-\varphi_{k_1}(f_1)-1]S_2\\
&&\ \geq 1+[(k'_1-k_1-1)Z_1]S_2\\
&&\ \geq 1.
\end{eqnarray*}

From \eqref{eq-type-I} and \eqref{eq-type-II} the last statement holds since the minimum integer of the integer set of $\mathbf{I}_2(k_1,f_1)$ minus the maximum integer of the integer set of $\mathbf{I}_1(s)$ is
\begin{eqnarray*}
&&1+[(k_1-1)Z_1+\varphi_{k_1}(f_1)-1+S_1]S_2-S_2-(S_1-1)S_2\\
&&\ =1+[(k_1-1)Z_1+\varphi_{k_1}(f_1)-1]S_2\\
&&\ \geq 1.
\end{eqnarray*}
\end{appendices}

\ifCLASSOPTIONcaptionsoff
  \newpage
\fi

\bibliographystyle{IEEEtran}
\bibliography{reference}

\end{document}